\documentclass[aps,pra,letterpaper,onrcolumn,nofootinbib,notitlepage,longbibliography]{revtex4-1}
\usepackage{amsfonts}
\usepackage{amsmath}
\usepackage{amssymb}
\usepackage{amsthm}
\usepackage{mathtools}
\usepackage{blkarray}
\usepackage{nicematrix}
\usepackage{enumerate}
\usepackage{bm}
\usepackage{dsfont}
\usepackage{graphicx}
\usepackage{nicefrac}
\usepackage{color}
\usepackage[all]{xy}
\usepackage{hyperref}
\usepackage{bbm}
\usepackage{comment}
\usepackage{tikz}
\usepackage{tikz-cd}
\usetikzlibrary {decorations.pathreplacing}
\usetikzlibrary{positioning}
\usetikzlibrary{calc}
\usepackage{tkz-graph}
\usetikzlibrary{arrows}
\usepackage{cancel}
\usepackage{xcolor}
\usepackage{booktabs}
\usepackage{multirow}
\usepackage{subcaption}
\usepackage{MnSymbol}
\usepackage[normalem]{ulem}

\usepackage[paperwidth=210mm,paperheight=297mm,centering,hmargin=2.1cm,vmargin=2.1cm]{geometry}

\newtheorem{theorem}{Theorem}
\newtheorem{proposition}{Proposition}
\newtheorem{corollary}{Corollary}
\newtheorem{lemma}{Lemma}
\newtheorem{definition}{Definition}

\newtheorem{obv}{Observation}
\newtheorem{rmk}{Remark}

\newcommand{\mf}[1]{#1}

\newcommand{\tr}{\mathrm{tr}}
\newcommand{\one}{\mathbbm{1}}
\newcommand{\diag}{\mathrm{diag}}
\newcommand{\id}{\mathrm{id}}

\newcommand{\Ad}{\mathrm{Ad}}

\newcommand{\UH}{\mc{U}(\cH)}

\newcommand{\SH}{\mc{S}(\cH)}
\newcommand{\LH}{\mc{L}(\cH)}
\newcommand{\LHsa}{{\LH_\mathrm{sa}}}

\newcommand{\ra}{\ensuremath{\rightarrow}}

\newcommand{\zz}{\mathbb Z}

\newcommand{\C}{\mathbb C}
\newcommand{\Z}{\mathbb Z}
\newcommand{\N}{\mathbb N}

\newcommand{\bi}{\mathbf{i}}
\newcommand{\bj}{\mathbf{j}}

\newcommand\mc[1]{\mathcal{#1}}

\newcommand{\cC}{\mc{C}}
\newcommand{\cD}{\mc{D}}

\newcommand{\cH}{\mc{H}}

\newcommand{\cO}{\mc{O}}

\newcommand{\cP}{\mc{P}}

\newcommand{\cU}{\mc{U}}

\newcommand{\SC}{\mc{SC}}

\newcommand{\suchthat}{\;\ifnum\currentgrouptype=16 \middle\fi|\;}

\begin{document}

\title{No quantum solutions to linear constraint systems from monomial measurement-based quantum computation in odd prime dimension}
\author{Markus Frembs}
\email{markus.frembs@itp.uni-hannover.de}
\affiliation{Institut f\"ur Theoretische Physik, Leibniz Universit\"at Hannover, Appelstraße 2, 30167 Hannover, Germany}
\author{Cihan Okay, Ho Yiu Chung}
\affiliation{Department of Mathematics, Bilkent University, Ankara, Turkey}

\begin{abstract}
    We combine the study of resources in measurement-based quantum computation (MBQC) with that of quantum solutions to linear constraint systems (LCS). Contextuality of the input state in MBQC has been identified as a key resource for quantum advantage, and in a stronger form, underlies algebraic relations between (measurement) operators which obey classically unsatisfiable (linear) constraints. Here, we compare these two perspectives on contextuality, and study to what extent they are related. More precisely, we associate a LCS to certain MBQC which exhibit strong forms of state-dependent contextuality, and ask if the measurement operators in such MBQC give rise to state-independent contextuality in the form of quantum solutions of its associated LCS. Our main result rules out such quantum solutions for a large class of MBQC. This both sharpens the distinction between state-dependent and state-independent forms of contextuality, and further generalises results on the non-existence of quantum solutions to LCS in finite odd (prime) dimension.
\end{abstract}
\maketitle

\vspace{-.75cm}

\section{Introduction}\label{sec: introduction}

Imagine you know the solution to some hard problem or computational task. You, the `prover', want to convince another party, the `verifier', of the fact that you really do know this solution. The complexity class $\mathrm{IP}$ of so-called interactive proofs represents the class of problems for which you are able to succeed in this task, and it corresponds with $\mathrm{IP=PSPACE}$, roughly the class of problems which are solvable with a deterministic Turing machine of polynomial-sized memory (running in possibly exponential time) \cite{Shamir1990}.\footnote{Since in the following we will not be concerned with complexity-theoretic matters, we refer the interested reader for precise definitions of (the definitions involved in) these statements to the cited references, specifically Ref.~\cite{JiEtAl2020}.} The situation changes if the verifier can not only interrogate a single, but multiple provers, and thus gains access to the correlations between their responses. In this case, of so-called multi-prover interactive proofs, $\mathrm{MIP}$, she can verify solutions to problems in $\mathrm{MIP=NEXP}\supsetneq\mathrm{PSPACE}$: non-deterministic, exponential time Turing machines \cite{BabaiEtAl1991}. Multi-prover interactive proofs can be interpreted as nonlocal games, which are winnable by provers who are not allowed to communicate, but may share a common (classical) strategy \cite{Mermin1990,Peres1991,CleveEtAl2004,Arkhipov2012,CleveMittal2014}.

As first demonstrated by Bell \cite{Bell1964,Bell1975}, and later verified experimentally \cite{FreedmanClauser1972,AspectEtAl1981,ZeilingerEtAl2015,ShalmEtAl2015}, access to entangled quantum states allows to outperform the best classical strategies in nonlocal games. This manifests in the fact that, once the provers are allowed to share entanglement, the otherwise analogously defined verification task defines the complexity class $\mathrm{MIP}^*$ which is dramatically more expressive than $\mathrm{MIP}$, as a verifier can now be convinced of any instance of a problem in the complexity class $\mathrm{MIP^*=RE}\supsetneq\mathrm{NEXP}$, the class of recursively enumerable languages, which includes the infamous halting problem \cite{JiEtAl2020}. This result, which among other things implies a negative resolution to Tsirelson's conjecture \cite{Tsirelson1980,Tsirelson2006} and the related Connes' embedding conjecture \cite{JungeEtAl2011,Fritz2012}, relies on the existence of linear constraint systems (systems of linear equations) over the ring $\zz_p$, which admit no classical solution (over $\zz_p$), but so-called `quantum solutions', which instead assign operators to variables, namely operators corresponding to measurements performed by nonlocal provers on a shared entangled state \cite{CleveMittal2014,CleveLiuSlofstra2017}. The study of quantum solutions to linear constraint systems thus lies within the remarkably profound intersection of quantum foundations, operator algebras and complexity theory.\footnote{This applies to the study of classical vs quantum solutions to more general constraint satisfaction problems, see e.g. \cite{AtseriasKolaitisSeverini2019,BulatovStanislav2025}.}

Examples of quantum solutions to binary, classically unsatisfiable linear constraint systems (LCS) can be constructed from the $n$-qubit Pauli group \cite{Mermin1990,Peres1991,Arkhipov2012,TrandafirLisonekCabello2022,MullerGiorgetti2025} (see Fig.~\ref{fig: Mermin(Peres) square and star} and Sec.~\ref{sec: MBQC and LCS}). What is more, known separation results establish the existence of quantum solutions to classically unsatisfiable LCS in the commuting operator paradigm, that is, based on possibly infinite-dimensional entanglement \cite{Slofstra2019,Slofstra2019b,SlofstraPaddock}. However, the glaringly large gap between ($n$-qubit) Pauli group-based constructions and infinite-dimensional quantum systems remains very little understood. In particular, at present, not a single example of a quantum solution on a finite-dimensional Hilbert space to a classically unsatisfiable LCS over $\zz_p$ with $p\neq 2$ odd is known.

Here, we approach this open problem from the perspective of contextuality (nonlocality) in measurement-based quantum computation (MBQC) \cite{RaussendorfBriegel2001}. More specifically, in Sec.~\ref{sec: MBQC and LCS} we recall that the Mermin-Peres square magic game \cite{Mermin1990,Peres1991} and Mermin's star \cite{Mermin1993} (see Fig.~\ref{fig: Mermin(Peres) square and star}) can be cast as MBQCs for which contextuality (nonlocality) results in a computational advantage over classical computation. MBQCs with this property are not restricted to qubit systems, but are readily constructed for arbitrary qudit systems (see Thm.~\ref{thm: AB qudit computation}), in which case they necessarily require measurement operators outside the Pauli group \cite{FrembsRobertsCampbellBartlett2022}. While the notion of contextuality (nonlocality) in MBQC inherently depends on the (entangled) resource state, in Thm.~\ref{thm: LCS = MBQC for Paulis} we prove that every MBQC based on Pauli observables constitutes a quantum solution to a LCS naturally associated to it (via Def.~\ref{def: associated LCS}). While this generalises MBQCs based on (generalisations of) the Mermin-Peres square and Mermin star to LCS over $\zz_p$ with $p$ odd (prime), the quantum solutions in the latter case turn out to be classical \cite{QassimWallman2020}. We then ask whether the generalised measurement operators involved in contextual MBQC for qudits contains quantum solutions to LCS over $\zz_p$. To this end, in Sec.~\ref{sec: from MBQC to LCS} we define a natural generalisation of the shift operators in the qudit Pauli group (and the Clifford hierarchy) to general permutations (see Def.~\ref{def: K-group}). The resulting operators naturally embed within the generalised stabiliser formalism based on so-called `monomial operators', studied in Ref.~\cite{VanDenNest2011,vanDenNest2013}. Our main result, Thm.~\ref{thm: FCO-2} shows that even within this larger group of unitaries, $\cP^{\otimes n}(p)<N^{\otimes n}_{SU(p)}(p)<SU(p)^{\otimes n}$, every quantum solution to a LCS over $\zz_p$ with $p$ odd prime reduces to a classical solution. Mathematically, this follows from the existence of a map $\phi:N^{\otimes n}_{SU(p)}(P)\ra\zz_p$ that restricts to a homomorphism in $p$-torsion abelian subgroups; its construction extends our earlier results in Ref.~\cite{FrembsChungOkay2022}.\footnote{As emphasised below, the notion of homomorphisms in $p$-torsion abelian subgroups is the relevant mathematical structure for the study of quantum solutions to (classically unsatisfiable) LCS.} Physically, it implies that state-dependent contextuality in MBQC \mf{(under natural assumptions laid out in Sec.~\ref{sec: from MBQC to LCS})} never lifts to state-independent contextuality for $p$ odd prime:
\begin{equation}\label{eq: mismatch state-(in)dep contextuality}
    \begin{tikzcd}[column sep=4cm]
        \mathrm{monomial\ \zz_p-MBQC\ (Def.~\ref{def: Z_p-MBQC},\ref{def: monomial Z_p-MBQC})} \arrow[rightarrow]{d}[swap]{\substack{(state-dependent)\\ contextuality}}
        \arrow[rightarrow,dashed]{r}{\mathrm{commutativity\ constraint}}
        & \mathrm{associated\ LCS\ (Def.~\ref{def: associated LCS})}
        \arrow[rightarrow]{d}{\substack{(state-independent)\\ contextuality}} \\
        \mathrm{nonclassical\ computation} 
        \arrow[rightarrow,"/"{anchor=center,sloped},red]{r}[yshift=0.1cm]{\mathrm{Thm.~}\ref{thm: no qsol from MBQC}}
        &\mathrm{quantum\ solution}
    \end{tikzcd}\tag{$*$}
\end{equation}

\section{Contextuality and MBQC}\label{sec: MBQC and LCS}

\subsection{Contextual MBQC for qubits}\label{sec: contextual MBQC for qubits}

In Def.~\ref{def: associated LCS} below, we will associate a linear constraint system (LCS) to every (deterministic, non-adaptive) measurement-based quantum computation (MBQC) with qudits of dimension $d$. To motivate this definition, it is instructive to first recall two well-known examples in the qubit case ($p=2$) - Mermin-Peres square and Mermin's star (see Fig.~\ref{fig: Mermin(Peres) square and star}) - as well as qudit generalisations thereof (see Sec.~\ref{sec: Contextual computation in qudit MBQC}). Throughout, we distinguish general (mostly odd) from odd prime integers in our notation by writing $d$ and $p$, respectively.\\

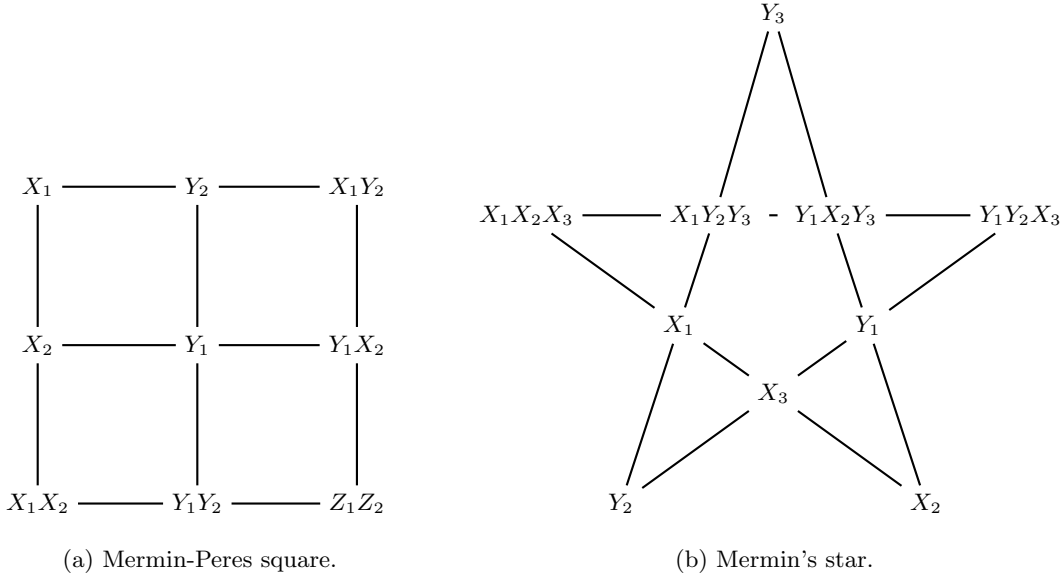
\begin{figure}[!htb]
    \centering
    \vspace{-0.5cm}
    \scalebox{1}{
    \hspace{-1cm}
    \begin{minipage}[b]{.4\textwidth}
        \centering
        %Mermin-Peres square
        \begin{tikzpicture}[x=0.07mm,y=0.07mm]
            %\draw [help lines,step=5mm] (0,0) grid (1000,1000);
            
            \node at (0,900)   (A) {$X_1$};
            \node at (300,900)   (B) {$Y_2$};
            \node at (600,900)   (C) {$X_1Y_2$};
            \node at (0,600)   (D) {$X_2$};
            \node at (300,600)   (E) {$Y_1$};
            \node at (600,600)   (F) {$Y_1X_2$};
            \node at (0,300)   (G) {$X_1X_2$};
            \node at (300,300)   (H) {$Y_1Y_2$};
            \node at (600,300)   (J) {$Z_1Z_2$};
            
            \draw[thick] (A) -- (B);
            \draw[thick] (B) -- (C);
            \draw[thick] (A) -- (D);
            \draw[thick] (B) -- (E);
            \draw[thick] (C) -- (F);
            \draw[thick] (D) -- (E);
            \draw[thick] (E) -- (F);
            \draw[thick] (D) -- (G);
            \draw[thick] (E) -- (H);
            \draw[thick] (F) -- (J);
            \draw[thick] (G) -- (H);
            \draw[thick] (H) -- (J);
        \end{tikzpicture}
        \vspace{0.1cm}
        \subcaption{Mermin-Peres square.}
    \end{minipage}%
    \begin{minipage}[b]{.5\textwidth}
        \centering
        %Mermin star
        \begin{tikzpicture}[x=0.065mm,y=0.065mm]
            %pentagram coordinates in 1000x1000 pixels: (A) (500,1000), (J) (191,0), (E) (1000,588), (B) (0,588), (K) (809,0)
            
            %\draw [help lines,step=5mm] (0,0) grid (150,50);
            \node at (500,1000)   (A) {$Y_3$};
            \node at (0,588)   (B) {$X_1X_2X_3$};
            \node at (382,588)   (C) {$X_1Y_2Y_3\ $};
            \node at (618,588)   (D) {$\ Y_1X_2Y_3$};
            \node at (1000,588)   (E) {$Y_1Y_2X_3$};
            \node at (309,363)   (F) {$X_1$};
            \node at (691,363)   (G) {$Y_1$};
            \node at (500,225)   (H) {$X_3$};
            \node at (191,0)   (J) {$Y_2$};
            \node at (809,0)   (K) {$X_2$};
            
            \draw[thick] (A) -- (C);
            \draw[thick] (A) -- (D);
            \draw[thick] (B) -- (C);
            \draw[thick] (C) -- (D);
            \draw[thick] (D) -- (E);
            \draw[thick] (B) -- (F);
            \draw[thick] (C) -- (F);
            \draw[thick] (D) -- (G);
            \draw[thick] (E) -- (G);
            \draw[thick] (F) -- (H);
            \draw[thick] (F) -- (J);
            \draw[thick] (G) -- (H);
            \draw[thick] (G) -- (K);
            \draw[thick] (H) -- (J);
            \draw[thick] (H) -- (K);
        \end{tikzpicture}
        \vspace{0.1cm}
        \subcaption{Mermin's star.}
    \end{minipage}}
    \vspace{0.2cm}
    \caption{Qubit operators (tensor products and identity operators omitted) used in the contextuality proofs in (a) the Mermin-Peres square, and (b) Mermin's star \cite{Mermin1990,Mermin1993}.}
    \label{fig: Mermin(Peres) square and star}
\end{figure}

\textbf{A LCS from Mermin-Peres square.} Let $\cP^{\otimes 
n}(d)=\langle X,Y,Z \rangle^{\otimes n}$ denote the $n$-qudit Pauli group. The nine two-qubit Pauli operators $x_k\in\cP^{\otimes 2}(2)$, $k\in\{1,\cdots,9\}$ in Fig.~\ref{fig: Mermin(Peres) square and star} (a) satisfy the following algebraic relations: (i) every operator is a square root of the identity $x_k^2 = \one \in \cP^{\otimes 2}(2)$ for all $k \in \{1,\cdots,9\}$, (ii) the operators in every row and column commute, that is, $[x_k,x_{k'}] = 0$ whenever $A_{ik} = 1 = A_{ik'}$ for some $i\in\{1,\cdots,6\}$, and (iii) the operators in every row and column obey the multiplicative constraints
\begin{equation*}
    \prod_{k=1}^9 x_k^{A_{ik}} = (-1)^{b_i}\one \quad\quad\forall i\in\{1,\cdots,6\}\; ,
\end{equation*}
where $A \in M_{6,9}(\zz_2)$ and $b \in \zz^6_2$ are defined as follows
\begin{equation}\label{eq: MP-square LCS}
    A = \begin{pmatrix}
         1 & 1 & 1 & 0 & 0 & 0 & 0 & 0 & 0  \\
         0 & 0 & 0 & 1 & 1 & 1 & 0 & 0 & 0  \\
         0 & 0 & 0 & 0 & 0 & 0 & 1 & 1 & 1  \\
         1 & 0 & 0 & 1 & 0 & 0 & 1 & 0 & 0  \\
         0 & 1 & 0 & 0 & 1 & 0 & 0 & 1 & 0  \\
         0 & 0 & 1 & 0 & 0 & 1 & 0 & 0 & 1  \\
    \end{pmatrix} \quad \quad \quad b = \begin{pmatrix}
         0 \\ 0 \\ 0 \\ 0 \\ 0 \\ 1
    \end{pmatrix} \; .
\end{equation}
Assume that there exists a \emph{noncontextual value assignment}, $\tilde{x}_k \in \zz^9_2$ with $x_k = (-1)^{\tilde{x}_k}\one$ such that
\begin{equation}\label{eq: LCS relation}
    \prod_{k=1}^9 x_k^{A_{ik}} = (-1)^{\sum_{k=1}^9 A_{ik}\tilde{x}_k}\one = (-1)^{b_i}\one\; ,
\end{equation}
Then the system of linear constraints $A\tilde{x}=b\mod 2$ admits a classical solution. Yet, note that no such classical solution exists: the Mermin-Peres square is an example of \emph{state-independent contextuality} \cite{Mermin1990,Peres1991}.

Generalising this example, one defines a \emph{linear constraint system (LCS)} `$Ax=b \mod d$' for any $A\in\mathrm{Mat}_{r,s}(\zz_d)$ and $b\in\zz^r_d$. A \emph{classical solution} to a LCS is a vector $\tilde{x}\in\zz_d^s$ solving the system of linear equations $A\tilde{x}=b\mod d$, whereas a \emph{quantum solution} is an assignment of unitary operators $x_k\in\UH$ such that (i) $x_k^d=\one$ for all $k\in [n]:=\{1,\cdots,n\}$ (`$d$-torsion'), (ii) $[x_k,x_{k'}] := x_kx_{k'}x_k^{-1}x_{k'}^{-1} = \one$ whenever $(k,k')$ appear in the same row of $A$, i.e., $A_{jk}\neq0\neq A_{jk'}$ for some $j\in[r]$ (`commutativity'), and (iii) $\prod_{k=1}^s x^{A_{jk}}=\omega^{b_j}\one$ for all $j\in [r]$, and where $\omega=e^{\frac{2\pi i}{d}}$ (`constraint satisfaction'). By identifying sets of commuting unitaries in a quantum solution with \emph{contexts},\footnote{More commonly, `contexts' are defined as subsets or subalgebras of commuting quantum observables, represented by self-adjoint operators. The present situation arises under the identification of self-adjoint operators with unitaries under exponentiation.} quantum solutions to LCS generalise the state-independent contextuality proof of Mermin-Peres square. For more details on linear constraint systems, see Ref.~\cite{Arkhipov2012,CleveMittal2014,CleveLiuSlofstra2017,Slofstra2019,QassimWallman2020,FrembsChungOkay2022}.\\

\textbf{Mermin's star and contextual MBQC with qubits.} Following Refs.~\cite{AndersBrowne2009,Raussendorf2013}, we encode the constraints in Mermin's star (Fig.~\ref{fig: Mermin(Peres) square and star} (b)) in computational form. Fix the three-qubit resource GHZ-state $|\psi\rangle=\frac{1}{\sqrt{2}}(|000\rangle+|111\rangle)$, and define local measurement operators $M_k(0)=X_k$, $M_k(1)=Y_k$, where $X_k,Y_k,Z_k\in \cP(2)$ denote the (local/single-qubit) Pauli measurement operators in Fig.~\ref{fig: Mermin(Peres) square and star} (b).\footnote{When it is clear from context, we will at times omit the subscript $k$ to avoid clutter.} Let $\mathbf{i} \in \zz_2^2$ be the input of the computation and define the local measurement settings by the (pre-processing) functions $f_1(\mathbf{i})=i_1$, $ f_2(\mathbf{i})=i_2$, and $f_3(\mathbf{i})=i_1+i_2\mod 2$. Further, the output of the computation is defined as $o=\sum_{k=1}^3 m_k\mod 2$, where $m_k\in\zz_2$ denotes the local measurement outcome of the $k$-th qubit. Note that the $f_k$ and $o$ are linear functions in the input $\mathbf{i}\in\zz_2^2$. It follows that Mermin's star defines a deterministic, non-adaptive $l2$-MBQC, where `deterministic' means that the output is computed with probability $1$, `non-adaptive' that the measurement settings $f_k=f_k(\mathbf{i})$ are functions of the input only (but not of previous measurement outcomes $m_k$), and `$l2$-MBQC' means that the classical side-processing (that is, pre-processing of local measurement settings via $f_k$ and post-processing of outcomes $m_k$ in the output function $o$) is restricted to $\zz_2$-linear computation (for details, see Ref.~\cite{HobanCampbell2011, Raussendorf2013,FrembsRobertsCampbellBartlett2022}).

Now, note that the resource state $|\psi\rangle$ is a simultaneous eigenstate of $X_1X_2X_3$, $-X_1Y_2Y_3$, $-Y_1X_2X_3$, and $-Y_1Y_2X_3$--- the three-qubit Pauli operators in the horizontal context of Mermin's star in Fig.~\ref{fig: Mermin(Peres) square and star} (b).\footnote{As in Fig.~\ref{fig: Mermin(Peres) square and star}, we use the common shorthand omitting tensor products in e.g. $X_1X_2X_3 := X_1 \otimes X_2 \otimes X_3$.} The output function $o: \zz_2^2 \rightarrow \zz_2$ is then the nonlinear OR gate. The nonclassical correlations in Mermin's star thus give rise to nonlinearity in the $l2$-MBQC, thereby boosting the classical computer beyond its (limited to $\zz_2$-linear side-processing) capabilities. In fact, nonlinearity becomes a witness of \emph{state-dependent contextuality} (nonlocality) in $l2$-MBQC more generally \cite{HobanCampbell2011,Raussendorf2013,FrembsRobertsBartlett2018,FrembsRobertsCampbellBartlett2022}.

We can capture the constraints of this $l2$-MBQC in the form of a LCS, too. Note that there are 10 observables in Mermin's star: $nd = 6$ local and $d^l=4$ nonlocal ones ($n=3$, $d=l=2$).\footnote{We will write $n$ for the number of qudit systems in an MBQC and $l$ for the length of its input dit-string $\bi\in\zz^l_d$.} Moreover, there are $d^l=4$ local constraints, each relating one global measurement operator to its $n$ local measurements, and one global constraint, relating the respective measurements. These correspond to the horizontal context in Fig.~\ref{fig: Mermin(Peres) square and star} (b). Taken together, this suggests to associate a LCS $Ax=b\mod 2$ to Mermin's star, given by
\begin{equation}\label{eq: M-star LCS}
     A =
        \begin{pNiceMatrix}[first-row]
            M(00) & M(01) & M(10) & M(11) & X_1 & Y_1 & X_2 & Y_2 & X_3 & Y_3 \\
            1 & 0 & 0 & 0 & 1 & 0 & 1 & 0 & 1 & 0 \\
            0 & 1 & 0 & 0 & 1 & 0 & 0 & 1 & 0 & 1 \\
            0 & 0 & 1 & 0 & 0 & 1 & 1 & 0 & 0 & 1 \\
            0 & 0 & 0 & 1 & 0 & 1 & 0 & 1 & 1 & 0 \\
            1 & 1 & 1 & 1 & 0 & 0 & 0 & 0 & 0 & 0
        \end{pNiceMatrix}
        \in M_{5\times 10}(\zz_2) \quad \quad \quad  b = \begin{pmatrix}
         0 \\ 0 \\ 0 \\ 0 \\ 1
    \end{pmatrix} \in \zz_2^5\; .
\end{equation}
Here, the first four rows correspond with the decomposition of (global) measurement operators into local ones
\begin{equation*}
    M(\bi)
    =\otimes_{k=1}^3 M_k(\bi) \quad\forall\bi\in\zz^2_2\; ,
\end{equation*}
whereas the last row represents the `parity constraint' in the horizontal context in Fig.~\ref{fig: Mermin(Peres) square and star} (b),
\begin{align}\label{eq: horizontal context M-star}
    \prod_{\bi\in\zz^2_2} M(\bi) = \omega^{\sum_{\bi\in\zz^2_2} o(\bi)}\one\; .
\end{align}
In particular, $b_\bi=0$ for all $\bi\in\zz^2_2$ and $b_5=\sum_{\bi\in\zz^2_2} o(\bi)\mod 2$, and for all but the last row, there is exactly one non-zero entry in every pair of $(4+(k-1)2+1,4+k2)$-th columns for all $k \in \{1,2,3\}$. \mf{Denoting these columns by $j_k \in \zz_2$ (and identifying $X_k=M_k(j_k=0)$ and $Y_k=M_k(j_k=1)$) we recover the measurement settings via $\delta_{f_k(\bi)j_k}$, that is, up to the last row, the $4+(k-1)2+(j_k+1)$-th column of $A$ is given by $f_k(\bi)=j_k$ for $\bi\in\zz^m_2$.}

It is easy to see that the LCS in Eq.~(\ref{eq: M-star LCS}) has no classical solution: subtracting the first four rows from the last yields a row of all zeros, and thus an unsatisfiable equation given the constraint $b_5=1$. This expresses the fact that (the MBQC based on) Mermin's star is contextual: it admits no noncontextual value assignment. Crucially, in general the notion of contextuality is state-dependent, as (global) measurements in MBQC are required to commute only on the subspace defined by the resource state $|\psi\rangle$ of the computation.

\subsection{Contextual MBQC for qudits}\label{sec: Contextual computation in qudit MBQC}

In this section we construct a (state-dependent) generalisation of the contextuality proof in Mermin's star \cite{Mermin1993} (see Fig.~\ref{fig: Mermin(Peres) square and star} (b)) from qubit to qudit systems in odd prime dimension. As in Sec.~\ref{sec: contextual MBQC for qubits}, the latter assumes a computational character within the scheme of measurement-based computation \cite{RaussendorfBriegel2001,BriegelRaussendorf2001,AndersBrowne2009}, which underlies one of the few known examples of provable quantum over classical computational advantage \cite{BravyiGossetKoenig2018,BravyiGossetKoenig2020}. The resource behind this advantage is contextuality, i.e., the non-existence of a noncontextual value assignment to all measurement operators. Sharp thresholds for state-dependent contextuality can be derived under the restriction of linear side-processing, sometimes denoted by $ld$-MBQC \cite{HobanCampbell2011,Raussendorf2013,FrembsRobertsBartlett2018,FrembsRobertsCampbellBartlett2022}.\\

\textbf{A contextual qudit example.} We generalise the contextual computation in Ref.~\cite{AndersBrowne2009,Raussendorf2013} based on the operators in Mermin's star to qudit systems of arbitrary odd prime dimension $p$. Define the following measurement operators by their action on computational basis states $\{|q\rangle\}_{q=0}^{p=1}$ and $\omega=e^{\frac{2\pi i}{p}}$, as follows:
\begin{equation}\label{eq: qudit measurement operators}
    M(\bi)|q\rangle
    :=\theta(\bi)\omega^{\bi q^{p-1}} |q+1\rangle, \quad \bi\in\zz_p
\end{equation}
Note first that $M(\bi)^p=\one$ if we set $\theta(\bi)^p=\omega^\bi$ and thus $\theta(\bi)^p\theta(\bj)^p\theta(p-\bi-\bj)^p=1$. A particular choice for $\theta(\bi)$ is given by $\theta(\bi)=e^{\frac{\bi 2\pi i}{p^2}}$. Similar to the qubit scenario in Sec.~\ref{sec: contextual MBQC for qubits}, we also specify $\zz_p$-linear functions, $f_1(\mathbf{i}):=i_1$, $f_2(\mathbf{i}):=i_2$, and $f_3(\mathbf{i}):=-i_1-i_2\mod p$ for $\mathbf{i}=(i_1,i_2)\in\zz_p^2$ and take the resource state to be given by
\begin{equation}\label{eq: qudit resource state}
    |\psi\rangle
    =\frac{1}{\sqrt{p}} \sum_{q=0}^{p-1}|q\rangle^{\otimes 3}\; .
\end{equation}

\begin{theorem}\label{thm: AB qudit computation}
    The $lp$-MBQC for $p$ prime with local measurements $M_k(f_k=f_k(\mathbf{i}))$ in Eq.~(\ref{eq: qudit measurement operators}), where the input $\mathbf{i}=(i_1,i_2)\in\zz_p^2$ defines measurement operators via the $\zz_p$-linear (pre-processing) functions $f_1(\mathbf{i}):=i_1$, $f_2(\mathbf{i}):=i_2$, and $f_3(\mathbf{i}):=-i_1-i_2\mod p$, is contextual when evaluated on the resource GHZ-state in Eq.~(\ref{eq: qudit resource state}).
\end{theorem}

\begin{proof}
    One easily computes the output function of the MBQC to be of the form:
    \begin{equation}\label{eq: general non-linear function}
        o(\mathbf{i}) = \begin{cases} 0 &\mathrm{if} \ i_1=i_2=0 \\
        1 &\mathrm{if} \ i_1 + i_2 \leq p \\
        2 &\mathrm{if} \ i_1 + i_2 > p \end{cases}
    \end{equation}
    By Thm.~1 in Ref.~\cite{FrembsRobertsBartlett2018}, the computation is contextual if $o(\mathbf{i})$ is at least of degree $p$. Note that the number of monomials (in the variables $i_1,i_2$) of degree at most $p-1$ is the same as the number of constraints for $i_1+i_2 \leq p-1$ in Eq.~(\ref{eq: general non-linear function}). If the computation is non-contextual, the latter will thus uniquely fix a function of degree at most $p-1$, namely $(i_1+i_2)^{p-1}$. However, $o$ does not satisfy this relation for $i_1+i_2\geq p$, that is, $o(\bi)\neq (i_1+i_2)^{p-1}$ and the output function must contain at least one term of degree $p$ or greater.
\end{proof}

Clearly, Thm.~\ref{thm: AB qudit computation} closely mimics the contextual computation in Sec.~\ref{sec: contextual MBQC for qubits}. It suggests that quantum solutions to (classically unsatisfiable) LCS may arise from (contextual) MBQC not only for $p=2$ but also for $p$ odd prime.\\

\textbf{Associated LCS.} As with Mermin's star in Eq.~(\ref{eq: M-star LCS}), we capture the constraints of the above \mf{$lp$-MBQC} in the form of a LCS. Note that there are $p^2$ operators $M(\bi)=\otimes_{k=1}^n M_k(\bi)$ for all $\bi\in\zz^2_p$, each built from $np$ local measurement operators with $n=3$. These satisfy $p^l$ \emph{locality constraints}, each relating one global measurement to its $n$ local ones, and one constraint relating the respective global measurement operators,
\begin{align}\label{eq: horizontal context qudit M-star}
    \prod_{\bi\in\zz^2_p} M(\bi) = \omega^{\sum_{\bi\in\zz^2_p} o(\bi)} \mathbbm{1}\; .
\end{align}
We therefore obtain a LCS $Ax=b \mod p$, in terms of the constraint matrix $A\in M_{(p^2+1)\times (p^2+3p)}(\zz_p)$ given by
\begin{equation}\label{eq: qudit M-star LCS}
     A =
        \begin{pNiceMatrix}[first-row]
            M(00) & M(01) & \cdots & M((p-1)(p-1)) & M_1(0) & \cdots & M_1(p-1) & M_2(0) & \cdots & M_3(p-1) \\
            1 & 0 & \cdots & 0 & \delta_{f_1(\bi=0),0} & \cdots & \delta_{f_1(\bi=0),p-1} & \delta_{f_2(\bi=0),0} & \cdots & \delta_{f_3(\bi=0),p-1} \\
            0 & 1 & \cdots & 0 & \delta_{f_1(\bi=1),0} & \cdots & \delta_{f_1(\bi=1),p-1} & \delta_{f_2(\bi=1),0} & \cdots & \delta_{f_3(\bi=1),p-1} \\
            0 & 0 & \ddots & 0 & \vdots & \ddots & \vdots & \vdots & \ddots & \vdots \\
            0 & 0 & \cdots & 1 & \delta_{f_1(\bi=p^2-1),0} & \cdots & \delta_{f_1(\bi=p^2-1),p-1} & \delta_{f_2(\bi=p^2-1),0} & \cdots & \delta_{f_3(\bi=p^2-1),p-1} \\
            1 & 1 & \cdots & 1 & 0 & 0 & 0 & 0 & 0 & 0
        \end{pNiceMatrix}\; ,
\end{equation}
as well as the constraint vector $b\in\zz^{p^2+1}_p$, given by $b_{\bi+1}=0$ for all $\bi\in\zz^2_p$ and $b_{p^2+1} = \sum_{\bi\in\zz^2_p} o(\bi)\mod p$. Again, the last row corresponds to the constraint between global measurement operators in Eq.~(\ref{eq: horizontal context qudit M-star}), whereas the remaining constraints correspond with the decomposition of global measurement operators into local ones. Moreover, for all but the last row, there is exactly one non-zero entry in every set of $p$ columns $p^l+(k-1)p+1,\cdots,p^l+kp$ for all $k \in [n]$. Denoting these columns by $j_k\in\zz_p$ we thus recover the measurement settings via $f_k(\bi)=j_k$.

As with the LCS associated to Mermin's star in Eq.~(\ref{eq: M-star LCS}), also the LCS in Eq.~(\ref{eq: qudit M-star LCS}) has no classical solution: multiplying each of the first $p^2$ rows by $p-1$ and adding them to the last yields a row of all zeros since $(p-1)\sum_{\bi\in\zz^2_p} \delta_{f_k(\bi),j_k}=p(p-1)=0 \mod p$ for all $j_k\in\zz_p$ since $l_k$ is a linear function. However, from Eq.~(\ref{eq: general non-linear function}),
\begin{align*}
    b_{2^p+1}
    = \sum_{\bi\in\zz^2_p} o(\bi) \mod p
    = \sum_{\bi\in\zz^2_p\mid 0<i_1+i_2 \leq p} + 2\sum_{\bi\in\zz^2_p \mid i_1+i_2 > p} \mod p
    = p-1\; .
\end{align*}
Motivated by the ongoing search for quantum solutions to classically unsatisfiable LCS over $\zz_p$ for $p$ odd (prime), this motivates the following question:
\begin{center}
    \emph{Are there contextual MBQC with qudits which define quantum solutions to their associated LCS?}
\end{center}
\mf{In the next section, we will turn this informal question into a well-defined mathematical problem, by considering a natural class of MBQC, and defining its associated notion of LCS. We then solve this problem in Sec.~\ref{sec: solution}, thereby providing a (negative) answer to the above question.}

\section{From (monomial) MBQC to LCS}\label{sec: from MBQC to LCS}

\subsection{$\zz_d$-MBQC and associated LCS}\label{sec: constraints on MBQC}

In this section, we generalise the examples of contextual MBQC in Sec.~\ref{sec: MBQC and LCS}, by associating a LCS $A(M)x=b(o) \mod d$ to every deterministic, non-adaptive MBQC $M$ with (local measurement) outcomes in $\zz_d$ more generally.

\begin{definition}\label{def: Z_p-MBQC}
    A \emph{$\zz_d$-MBQC, with $n$-qudit resource state $|\psi\rangle\in\mathbb{CP}^{nd-1}$}, is defined by the relation
    \begin{equation}\label{eq: output function in MBQC}
        M(\bi)|\psi\rangle
        =\otimes_{k=1}^n M^{c_k(\bi)}_k(\bi)|\psi\rangle
        = \omega^{o(\bi)}|\psi\rangle\; ,
    \end{equation}
    where $M_k(\bi)\in U(d)$ with $M^d_k(\bi)=\one_k$ and $M(\bi)\in U^{\otimes n}(d)$ denote the local, respectively global `measurement operator' defined by (pre-processing) functions $M_k:I\ra U(d)$ and $c_k:I\ra\zz_d$ for all $k\in[n]$ with respect to a fixed set of inputs $I$, and output function $o:I\rightarrow\zz_d$ given by $o(\bi)=\sum_{k=1}^n c_k(\bi)m_k\mod d$ for all $\bi\in I$.
\end{definition}

Note that the measurement outcomes $m_k\in\zz_d$ of $M_k$ are individually random, yet their post-processed sum is deterministic by assumption. Next, we associate a LCS to every $\zz_d$-MBQC $M$ as follows.

\begin{definition}\label{def: associated LCS}
    Let $M$ be a $\zz_d$-MBQC with output function $o:I\ra\zz_d$ as in Def.~\ref{def: Z_p-MBQC}. Then the \emph{linear constraint system $A(M)x = b(o)\mod d$ associated with $M$} is defined, for all inputs $\bi,\bj\in I$ of cardinality $|I|$, by
    \begin{equation}\label{eq: associated TG-LCS}
     A(M):=
        \begin{pNiceMatrix}[first-row]
            M(\bj) & M_k(\bj) \\
            \delta_{\bi,\bj} & c_k(\bi)\delta_{\bi,\bj} \\
            (1^{|I|})^T & (0^{n|I|})^T \\
        \end{pNiceMatrix} \in M_{(|I|+1) \times (|I|+n|I|)}(\zz_d) \hspace{1.25cm}
    b(o):= 
        \begin{pNiceMatrix}
            0^{|I|} \\ 
            \sum_{\bi\in I} o(\bi)\mod d
        \end{pNiceMatrix}
        \in\zz^{|I|+1}_d\; .
    \end{equation}
\end{definition}

Clearly, Def.~\ref{def: associated LCS} generalises Eq.~(\ref{eq: M-star LCS}) and Eq.~(\ref{eq: qudit M-star LCS}). We add some further comments on Def.~\ref{def: Z_p-MBQC} and Def.~\ref{def: associated LCS}.

\begin{itemize}    
    \item \textbf{determinism.} $\zz_d$-MBQC is chosen to be deterministic in order to compare with (deterministic) LCS.

    \item \textbf{non-adaptive MBQC.} By Def.~\ref{def: Z_p-MBQC}, the input set $I$ is fixed prior to the computation, consequently it does not include previous measurements outcomes, and thus makes $\zz_d$-MBQC non-adaptive. This is natural from the perspective of (quantum solutions to) LCS, which merely require commutativity of operators, and thus express simultaneous measurability without any temporal order. Despite this stark restriction of MBQC, it is worth pointing out that the setting of non-adaptive MBQC exhibits strong forms of (state-dependent) contextuality (see  Refs.~\cite{HobanCampbell2011,Raussendorf2016,FrembsRobertsBartlett2018,FrembsRobertsCampbellBartlett2022}).

    \item \textbf{classical side-processing.} Eq.~(\ref{eq: associated TG-LCS}) poses no restriction on classical side-processing, in fact, Def.~\ref{def: Z_p-MBQC} (and Def.~\ref{def: associated LCS}) do not specify any algebraic structure on inputs (nor post-processing on measurement outcomes $m_k$). This is in contrast to the linear side-processing in the examples in Sec.~\ref{sec: MBQC and LCS},\footnote{We reiterate that this restriction is of little practical value, but mainly serves to magnify the resourcefulness of contextuality.\label{fn: linear side-processing}} which imposes additional constraints on the maps $M_k:I\ra U(d)$, specifying local measurement operators in Eq.~(\ref{eq: output function in MBQC}).\footnote{A common choice is $M_k(\bi)=U(\bi)M_kU^\dagger(\bi)$ \cite{Raussendorf2016,FrembsRobertsBartlett2018,FrembsRobertsCampbellBartlett2022} with $U(\bi)$ a (projective) representation of the group $I=\zz^l_d$ and $l\in\mathbb{N}$.}

    \item \textbf{local symmetry.} Note that a generic set of linear equations is invariant under arbitrary permutations of its rows and columns. The LCSs associated with $\zz_d$-MBQC in Eq.~(\ref{eq: associated TG-LCS}) have a more refined invariance property due to the local structure of measurement operators. Symmetries of LCS in the form of Eq.~(\ref{eq: associated TG-LCS}) must respect the decomposition into (local) subsystems in the tensor product decomposition $\cH=\otimes_{k=1}^n\cH_k$ of measurement operators, and are thus constrained to be of the form $S_{|I|}\times S_n$.\footnote{Notably, this symmetry group is generally not compatible with the restriction to linear side-processing in $ld$-MBQC.$^{\ref{fn: linear side-processing}}$}
\end{itemize}

Def.~\ref{def: associated LCS} associates a LCS to $\zz_d$-MBQC with qudits in arbitrary dimension and with general (local) measurement operators, and thus establishes a relationship between these two concepts on a general level. In the following, we will only consider a particular class of (deterministic, non-adaptive) $\zz_d$-MBQC.

\subsection{Monomial MBQC}

In this section, we further restrict the class of MBQC to be studied below. More precisely, we will restrict the type of measurement operators in $\zz_d$-MBQC. There are several reasons for imposing such a restriction.

First, we note that the construction of contextual MBQC in Ref.~\cite{FrembsRobertsBartlett2018} only requires unitary operators that are semi-Clifford (see Sec.~\ref{sec: relation with Clifford group}), and thus are of a very special form. Second, while Def.~\ref{def: Z_p-MBQC} remains agnostic about the (classical) pre-processing of measurement operators, it is natural to demand certain compatibility conditions between classical side-processing and the type of map $I\ra U^{\otimes n}(d)$ (cf. Ref.~\cite{Raussendorf2013}). What is more, when going over to the general (adaptive) setting, this should be accompanied by computational constraints that allow for efficient `feed-forward', that is, updating of future measurements by accounting for `by-products' \cite{RaussendorfBriegel2001}. A third reason is more economical in nature. Since the search of quantum solutions to LCS in odd (prime) dimensions remains a challenging open problem, a reasonable approach is to study restricted instances of this problem, and a natural restriction is to consider quantum solutions that take values in proper subgroups of the full unitary group only. From this perspective, allowing for general measurement operators in $\zz_d$-MBQC is a rather weak restriction - especially when compared with the fact that contextual MBQC requires unitaries of a highly restricted form only Ref.~\cite{Raussendorf2013,Raussendorf2016,FrembsRobertsBartlett2018,FrembsRobertsCampbellBartlett2022}.

Given the above considerations, a natural class of $\zz_d$-MBQC arises from considering measurement operators within the monomial stabiliser formalism developped in Ref.~\cite{VanDenNest2011,vanDenNest2013}.\footnote{Another natural restriction is to consider Clifford unitaries only. We will analyse this case in Thm.~\ref{thm: Clifford noncontextual} in Sec.~\ref{sec: relation with Clifford group}.}

\begin{definition}\label{def: monomials}
    A unitary $U\in U(d)$ is called \emph{monomial (with respect to the basis $\{|q\rangle\}_{q\in\zz_d}$) of $\cH=\C^d$} if
    \begin{align*}
        U=\Pi_\sigma S_\xi\; ,
    \end{align*}
    where \mf{$S_\xi=\mathrm{diag}(\xi(0),\cdots,\xi(d-1))$, $\xi:\zz_d\ra U(1)$} is diagonal in the basis (and called a \emph{generalised phase gate}), and $\Pi_\sigma$ is the matrix representation of the permutation $\sigma\in S_d$ acting on the elements in the basis.
    
    A state $|\Psi\rangle\in\mathbb{CP}^{d-1}$ is called an \emph{$M$-state} if it is the common $+1$-eigenspace of a set of monomial unitaries.
\end{definition}

Equivalently, a unitary is monomial if it has a single non-zero entry in every row and column. It is easy to see that monomial unitaries form a proper subgroup of the unitary group $N_{U(d)}(d)\leq U(d)$ (see Eq.~(\ref{ses: N})).

\begin{definition}\label{def: monomial Z_p-MBQC}
    A $\zz_d$-MBQC is \emph{monomial}, if its resource state is a generalised stabiliser state (`$M$-state') and measurement operators are monomial unitaries with respect to the product computational basis (see Def.~\ref{def: monomials}).
\end{definition}

Indeed, restricting resource states to generalised stabiliser states, so-called `$M$-states', and measurements to monomial unitaries, ensures both that updated resource states remain $M$-states and allows for efficient feed-forward. Notably, unlike the heavily restricted stabiliser formalism, its monomial generalisation encompasses many cases of physical interest (as discussed in detail in Ref.~\cite{VanDenNest2011}). This includes, in particular, the setting of (contextual) MBQC studied in Ref.~\cite{FrembsRobertsCampbellBartlett2022}.\footnote{It is also worth mentioning that Def.~\ref{def: monomial Z_p-MBQC} does not impose a restriction to linear side-processing, as is often found in previous work on (the resourcefulness of) contextuality in MBQC \cite{AndersBrowne2009,Raussendorf2013,FrembsRobertsBartlett2018,FrembsRobertsCampbellBartlett2022}.} The latter establishes a link between the degree of contextuality (defined via the degree of the output function in Def.~\ref{def: Z_p-MBQC}) and the (semi-)Clifford hierarchy. For this reason, it useful to contrast monomial unitaries with the Clifford group and hierarchy in more detail.

\subsection{Comparison with Clifford group and hierarchy}\label{sec: relation with Clifford group}

We compare the group of local measurement operators in monomial $\zz_d$-MBQC (see Def.~\ref{def: monomial Z_p-MBQC}) with the Clifford group and hierarchy. Recall that the \emph{$n$-qudit Clifford group}, $\cC^n(d)=N_{\cU(\C^{d^n})}(\cP^n(d))$ is the normaliser of the Pauli group $\cP^n(d)$. The \emph{Clifford hierarchy} is defined recursively by
\begin{align*}
    \cC^n_m(d)
    &:=\{U\in\cU(\C^{d^n})\mid UPU^{-1}\in\cC^n_{m-1}(d)\ \forall P\in\cP^n(d)\} &
    \cC(d)
    &:=\bigcup_{m=1}^\infty\cC^n_m(d)\; ,
\end{align*}
with $\cC^n_1(d)=\cP^n(d)$ and $\cC^n_2(d)=\cC^n(d)$. $(\cC_m(d))^{\otimes n}\leq\cC^n_m(d)$ is called the \emph{local Clifford group}. Note that for $m\geq 3$ (and $n,d$ arbitrary), $\cC^n_m(d)$ is not a group and its characterisation remains an open problem. Under the restriction to \emph{diagonal elements} (in the eigenbasis of the Pauli-Z operator), one obtains a restricted hierarchy of groups $\cD^n_m(d)$ which has been characterised in Ref.~\cite{CuiGottesmanKrishna2017}. The hierarchy of \emph{semi-Clifford operators} \cite{ZengChenChuang2008} is defined by $\SC^n_m(d)=\{U\in\cC^n_m(d)\mid \exists C,C'\in\cC^n(d):\ CUC'\in\cD^n_{m-1}(d)\}$, and (together with $\cC^n_m(d)$) plays a key role in fault-tolerant quantum computation via gate teleportation protocols with magic states \cite{GottesmanChuang1999,ZhouLeungChuang2000,ZengChenChuang2008,BeigiShor2010,deSilva2021,ChenDeSilva2024}.

In the case of qudits of odd dimension, the resourcefulness of magic states has been identified with contextuality \cite{howard2012qudit}. Moreover, contextuality also acts as a resource in the architecture of measurement-based quantum computation \cite{RaussendorfBriegel2001,BriegelRaussendorf2001}. In particular, in Ref.~\cite{FrembsRobertsCampbellBartlett2022} it is shown that under the restrictions to linear side-processing and non-adaptivity, computing a function of degree $D$ generally necessitates measurement operators that are $D$-th level, local semi-Clifford. The measurement operators in (contextual) MBQC thus seem natural candidates in the search for quantum solutions to linear constraint systems over $\zz_d$ with $d$ odd.

Yet, since quantum solutions to LCS consist of unitaries of $d$-torsion not all Clifford operators qualify. For instance, within the generating set of the single-qudit Clifford group (see e.g. Ref.~\cite{deBeaudrap2013}),
\begin{align}\label{eq: Clifford generators}
    S&=\sum_{q=0}^{d-1}\tau^{q^2}|q\rangle\langle q| &
    F&=\frac{1}{\sqrt{d}}\sum_{q,q'=0}^{d-1}\omega^{qq'}|q'\rangle\langle q| &
    M_a&=\sum_{q=0}^{d-1}|aq\rangle\langle q|\quad\mathrm{for}\ a\in\zz^*_d\; ,
\end{align}
where $\tau^2=\omega=e^{\frac{2\pi i}{d}}$ and $|q'\rangle$ is an eigenstate of $X$, that is, $X|q'\rangle=\omega^{q'}|q'\rangle$, no element has order $p$ for $d=p$ odd prime. What is more, $(\cC(p))^{\otimes n}$ already exhausts all permutation gates in the local Clifford hierarchy.\footnote{Since SWAP is a Clifford gate, the full (non-local) Clifford group $\cC^n$ also contains arbitrary permutations between tensor factors, which includes permutations of $p$-torsion. Here, we do not consider such permutations.}

\begin{lemma}\label{lm: Clifford permutations}
    Let $U=U_\sigma$ for $\sigma\in S_p$ for $p$ odd prime be the permutation matrix on the computational register, $U_\sigma|q\rangle=|\sigma(q)\rangle$ (with $\{|q\rangle\}_{q=0}^{p-1}$ defined via $Z|q\rangle=\omega^q|q\rangle$). Then $U\in\cC_m(p)$, $m\geq 2$ if and only if $U\in\cC_2(p)=\cC(p)$.
\end{lemma}

\begin{proof}
    Since $\cC_2(p)\leq\cC_m(p)$ for $m\geq 2$ one direction is obvious. For the other direction, consider the operator $U_\sigma XU^\dagger_\sigma=U_\sigma U_tU^{-1}_\sigma=U_{\sigma t\sigma^{-1}}$, where $U_t=X$ and $t\in S_p$ denotes the shift permutation $t.q=q+1\mod p$. It is well known that $U_\sigma\in\cC(p)$ if and only if $U_\sigma$ is an affine permutation, that is, of the form $\sigma=t_{a,b}$ where $t_{a,b}.q=aq+b\mod p$ \cite{deBeaudrap2013}. In particular, compared to the Pauli group ($m=1$), the Clifford group ($m=2$) contains the elements $U_{t_a}=\sum_{q=0}^{p-1}|aq\rangle\langle q|$ for $a\in\zz^*_p$. Next, let $U_\sigma\in\cC_3(p)$ and consider the condition $U_\sigma XU^\dagger_\sigma\in\cC_2(p)$. From above, $U_\sigma XU^\dagger_\sigma=U_{t_{a,b}}$. Yet, since $t_{a,b}$ has order a multiple of $a$ with $a|p-1$, while $X$ has order $p$, this implies $a=1$. The constraint is thus the same as in the second level, and thus $U_\sigma\in\cC_2(p)$.
\end{proof}

While the measurement operators in monomial $\zz_p$-MBQC do not include all (local) Clifford operators $\cC_m(p)$ for $m>2$ (since $F\notin N(Z)$), it turns out that they do contain all (local) semi-Clifford $p$-torsion elements.

\begin{lemma}\label{lm: d-torsion semi-Cliffords}
    Let $U\in\SC_m(p)$ for some $m\in\zz$ and $p$ prime. If $U^p=\one$, then $U\in \mf{K}=T(Z)\rtimes\zz_p$.
\end{lemma}

\begin{proof}
    By definition (see Ref.~\cite{,ZengChenChuang2008}), $U$ being semi-Clifford means that it maps a maximal diagonal subgroup of $\cP^n(d)$ to another maximal abelian subgroup of $\cP^n(d)$. Wlog, we may write $U=VD$, where $D\in\cD_m(p)$ for some $m$ and $V$ acts on the lattice of abelian subgroups of $\cP^n(d)$ by conjugation, explicitly it maps $D$ to the conjugate subgroup $D_V=\Ad_V(D)=VDV^{-1}$. Since for every $D_{V^i}\in V^i\cD_m(p)V^{-i}$ there exists $D(i)\in\cD_m(p)$ such that $D(i)=V^{-i}D_{V^i}V^i$, $p$-torsion implies $U^p=(\Ad_V\circ D)^p=\Ad^p_V\circ\prod_{i=0}^{p-1}D(i)=\one$ which therefore requires that both $\Ad^p_V=\mathrm{id}$ and $\prod_{i=0}^{p-1}D(i)=\one$. Under the identification of Pauli operators (up to phase) with vectors $v=(x,z)\in\zz^{2n}_p$ given for $n=1$ by $W(v)=\tau^{-xz}Z^zX^x$, the action of $\Ad_V$ is given by a symplectic transformation $S_V\in\mathrm{Sp}(2,\zz_p)$ \cite{deBeaudrap2013}. As a consequence, we must have $S^p_V=\one$, equivalently $(S_V-\one)^p=0$, that is, $S_V-\one$ is a nilpotent matrix and thus, up to conjugation, of the form $S_V-\one=\begin{pmatrix}0& c\\ 0& 0\end{pmatrix}$ for $c\in\zz_p$. Consequently, $S_V$ is an affine transformation, hence, $U$ factorises as an affine permutation and a diagonal matrix, and thus $U\in K$.
\end{proof}

As a consequence, we find that deterministic, non-adaptive $\zz_p$-MBQC restricted to the local Clifford group in odd prime dimension $p$ does not give rise to quantum solutions of its associated LCS.

\begin{theorem}\label{thm: Clifford noncontextual}
    Let $\Gamma$ be the solution group of a LCS over $\zz_p$ for $p$ odd prime. Then the LCS admits a quantum solution $\eta:\Gamma\ra \cC^{\otimes n}(p)$ if and only if it admits a classical solution $\eta:\Gamma\ra\zz_p$.
\end{theorem}

\begin{proof}
    Since $\SC_2(p)=\cC(p)$, this follows from Lm.~\ref{lm: d-torsion semi-Cliffords}, together with Thm.~2 in Ref.~\cite{FrembsChungOkay2022}.
\end{proof}

Thm.~\ref{thm: Clifford noncontextual}, Lm.~\ref{lm: Clifford permutations} and the fact that $\cC_m(d)$ and $\SC_m(d)$ do not form groups for $m>2$ suggests to consider different extensions of the qudit Pauli group. Note that for $d$ odd the Pauli group coincides with the Heisenberg-Weyl group, $\cP(d)\cong H(d)$ which is a semi-direct product of phase and shift operators (see Eq.~(\ref{ses: H})), and where the latter act as automorphisms on the former (see Eq.~(\ref{ses: H})). This perspective suggests a generalisation from shifts to more general permutations (of order $d$). The group $L(d)$ (see Eq.~(\ref{ses: L}) and App.~\ref{app: comparison}) generalises the Heisenberg-Weyl group to arbitrary permutations.\footnote{Since for $d=2$, $\langle X\rangle\cong S_2$, $L(d)$ may be seen as an alternative generalisation of the qubit Pauli group to qudits.} More generally, the group of monomial matrices $N(d)$ is the group that arises by generalising the group $K$ (see Def.~\ref{def: K-group} below),\footnote{$K$ is the group generated by conjugating the Pauli-$X$ operator with generalised phase gates, e.g. in $\cD_m(d)$ \cite{FrembsRobertsCampbellBartlett2022,FrembsChungOkay2022}.} to arbitrary permutations.

\section{Solutions of LCS associated to $\zz_p$-MBQC}\label{sec: solution}

\subsection{Noncontextual $\zz_p$-MBQC and classical solutions to its associated LCS}

We first relate noncontextuality of a $\zz_p$-MBQC with classical solvability of its associated LCS.

\begin{lemma}\label{lm: classical solutions}
    Let $M$ be noncontextual $\zz_d$-MBQC. Then its associated LCS has a classical solution.
\end{lemma}

\begin{proof}
    $M$ being noncontextual means there exists a value assignment $v:\mc{O}\ra\zz_d$ with $\mc{O}=\{M(\bi)\mid\bi\in I\}\cup\{M_k(\bi)\mid k\in[n],\bi\in I\}$ for $|\mc{O}|=(n+1)|I|$
    such that $o(\bi)=v(M(\bi))=\sum_{k=1}^n v(M_k(\bi))\mod d$ for all $\bi\in I$. Clearly, $v\in\zz^{(n+1)|I|}_d$ therefore is a solution to the associated LCS $A(M)v=b(o)\mod d$ in Def.~\ref{def: associated LCS}.
\end{proof}

However, note that the existence of a classical solution to the LCS $Ax=b\mod d$ associated to a $\zz_d$-MBQC does not automatically imply that it is noncontextual. There are two reasons for this:

First, a classical solution may not lift to a value assignment for the measurement operators in a $\zz_d$-MBQC since its associated LCS in Eq.~(\ref{eq: associated TG-LCS}) may not encode all algebraic relations between (subsets of commuting) operators that constitute a quantum solution to it. For instance, recall from Ref.~\cite{FrembsRobertsBartlett2018} that a deterministic, non-adaptive $lp$-MBQC for $d=p$ odd prime is contextual if its output function has degree $\deg(o)>p$. Yet, as long as $\sum_{\bi\in\zz^l_p}o(\bi)=0 \mod p$, which may be the case even for $\deg(o)>p$, the LCS associated to it is classically solvable. Second, the obvious candidate value assignment $v:\cO\ra\zz_d$ defined by $v(O_i)=x_i$ for all $O_i\in\cO$ in the $\zz_d$-MBQC may not obey the spectral constraint, $v(O)\notin\mathrm{sp}(O)$, that is, $v(O)$ may not be an eigenvalue of $O$. Clearly, this constraint is satisfied if all measurement operators $\cO$ have spectrum equal to $\zz_d$.

Despite these differences between (non)contextual $\zz_d$-MBQC and the existence of classical solutions to its associated LCS, contextuality of an $\zz_d$-MBQC still indicates that its associated LCS is not classically solvable.

\begin{lemma}\label{lm: contextual solution}
    For every function $o: \zz^l_p\ra\zz_p$ with $\sum_{\bi\in\zz^l_p}o(\bi)\neq 0 \mod p$ and $p$ prime there exists a deterministic, non-adaptive $\zz_p$-MBQC implementing this function such that its associated LCS has no classical solution.
\end{lemma}

\begin{proof}
    For any given output function $o:\zz^l_p\ra\zz_p$, Lm.~1 in Ref.~\cite{FrembsRobertsCampbellBartlett2022} constructs a deterministic, non-adaptive $lp$-MBQC $M$ with linear pre-processing functions $f_k:\zz^l_p\ra\zz_p$ (and trivial post-processing) that implements $o$ by action on the $n$-qudit GHZ state $|\Psi\rangle=\frac{1}{\sqrt{p}}\sum_{q=0}^{p-1} \otimes_{k=1}^n|q\rangle_k$ for some $n\leq p^l$ with measurement operators in $N^{\otimes n}_{SU(p)}(p)$. Consider the constraint matrix of its associated LCS in Eq.~(\ref{eq: associated TG-LCS}). By multiplying each of the first $p^l$ rows by $p-1$ and adding these to the last row, we obtain a row of only zeros. Yet, since $\sum_{\bi\in\zz^l_p} o(\bi)\neq 0\mod p$ by assumption, the LCS $A(M)x=b(o)$ has no classical solution.
\end{proof}

\subsection{$\zz_d$-MBQC with Pauli measurements}\label{sec: MBQC as quantum solutions to LCS for Pauli observables}

We observe the following general relation between $\zz_d$-MBQC and quantum solutions of its associated LCS.

\begin{lemma}\label{lm: MBQC + comm = LCS}
    A $\zz_d$-MBQC $M$ defines a quantum solution to its associated LCS in Def.~\ref{def: associated LCS} if and only if the (global) measurement operators $\{M(\bi)\}_{\bi\in I}$ pairwise commute.
\end{lemma}

\begin{proof}
    This follows immediately from Def.~\ref{def: associated LCS}. More precisely, $d$-torsion holds since local (hence, global) measurement operators in $\zz_d$-MBQC satisfy $M^d_k(\bi)=\one_k$. Commutativity of operators in rows labelled by $\bi\in I$ follows from the tensor product decomposition $M(\bi)=\otimes_{k=1}^n M^{c_k(\bi)}_k(\bi)$, and holds by assumption for the global measurement operators in the $|I|+1$-th row. $M$ thus defines a quantum solution to its associated LCS.
\end{proof}

Lm.~\ref{lm: MBQC + comm = LCS} therefore suggests a path for constructing quantum solutions to LCS from $\zz_d$-MBQC. Note however that unlike in Mermin's star in Fig.~\ref{fig: Mermin(Peres) square and star} (b), whose three-qubit Pauli operators commute in every context---including the horizontal one---measurement operators in ($\zz_d$-)MBQC generally only commute on the resource state, as is the case in the examples in Sec.~\ref{sec: Contextual computation in qudit MBQC}.\footnote{As such, deterministic, non-adaptive $\zz_d$-MBQC naturally only gives rise to \emph{operator solutions} of LCS (without the commutativity constraint). These are classical if and only if they define a local hidden variable model for the $\zz_d$-MBQC.} Still, the commutativity constraint may hold in certain instances of contextual $\zz_d$-MBQC and thus give rise to quantum solutions of its associated LCS also for $d$ odd.

In order to analyse such cases, we start with the following straightforward observation.

\begin{lemma}\label{lm: commutativity up to phase is state-independent}
    Let $a,b\in\UH$ with $[a,b]=aba^{-1}b^{-1}=z\one$ for $z\in U(1)$, and $\rho\in\SH$ with $\tr[\rho ab]\neq 0$.\footnote{Here, $\SH=\{\rho\in\LHsa\mid\rho>0,\tr[\rho]=1\}$ denotes the set of quantum states, represented in terms of density matrices.} Then
    \begin{align}
        [a,b]_\rho
        :=\tr[\rho ab]-\tr[\rho ba]=0 \quad \Longleftrightarrow \quad [a,b]=0\; .
    \end{align}
\end{lemma}

\begin{proof}
    Clearly, if $[a,b]=0$, then also $[a,b]_\rho=\tr[\rho ab]-\tr[\rho ba]=\tr[\rho ba]-\tr[\rho ba]=0$. Conversely, we have
    \begin{align*}
        0=[a,b]_\rho
        =\tr[\rho ab]-\tr[\rho ba]
        =z\tr[\rho ba]-\tr[\rho ba]
        =(z-1)\tr[\rho ba]\; ,
    \end{align*}
    from which it follows that $z=1$, since we assumed $\tr[\rho ab]\neq 0$, hence, $[a,b]=\one$ equivalently $ab=ba$.
\end{proof}

Lm.~\ref{lm: commutativity up to phase is state-independent} naturally applies to the $n$-qudit Pauli group $\cP^{\otimes n}$, since $[\cP^{\otimes n},\cP^{\otimes n}]=Z(\cP^{\otimes n})=\{\omega^c\one\mid c\in\zz_d\}$ where $\omega=e^{\frac{2\pi i}{d}}$ as before. In particular, if two Pauli operators share a common eigenstate then they also commute.

\begin{theorem}\label{thm: LCS = MBQC for Paulis}
    Every $\zz_d$-MBQC $M$ with output function $o:I\ra\zz_d$ (with $d$ arbitrary) and restricted to Pauli measurements defines a quantum solution to its associated LCS $A(M)x=b(o)\mod d$ in Def.~\ref{def: associated LCS}.
\end{theorem}

\begin{proof}
    This follows immediately from Lm.~\ref{lm: MBQC + comm = LCS} and Lm.~\ref{lm: commutativity up to phase is state-independent}.
\end{proof}

In the qubit case, Thm.~\ref{thm: LCS = MBQC for Paulis} generalises the examples in Sec.~\ref{sec: contextual MBQC for qubits}, as well as related results \cite{Arkhipov2012,TrandafirLisonekCabello2022,MullerGiorgetti2025}. In the qudit case with $d$ odd, it relates the fact that LCS with quantum solutions in the $n$-qudit Pauli group are classically solvable \cite{QassimWallman2020}, with the fact that MBQC restricted to Pauli measurements is noncontextual \cite{Gross2006}.

\begin{corollary}\label{cor: qudit Pauli MBQC}
    Every $\zz_d$-MBQC with $d$ odd and restricted to Pauli measurements is noncontextual.
\end{corollary}

\begin{proof}
    By Thm.~\ref{thm: LCS = MBQC for Paulis}, a $\zz_d$-MBQC $M$ constitutes a quantum solution to its associated LCS, which admits a classical solution by Thm.~2 in Ref.~\cite{QassimWallman2020}. Since Pauli operator have spectrum $\zz_d$, every classical solution assigns valid spectral values. Hence, by extending the LCS associated with $M$ to one which contains all product constraints between Pauli observables, Ref.~\cite{QassimWallman2020} in particular assures the existence of a classical solution which defines a valid value assignment for the full $n$-qudit Pauli group.
\end{proof}

Next, we turn to LCS with quantum solutions in the full group of measurement operators in $\zz_d$-MBQC.

\subsection{Quantum solutions to LCS associated with monomial $\zz_p$-MBQC}

Since every LCS with a quantum solution in the $n$-qudit Pauli group for $d$ odd is classically solvable \cite{QassimWallman2020}, we extend our search to unitary operators beyond the Pauli group. While quantum solutions can ultimately be built from arbitrary unitary operators, here, we consider a restricted version of the problem: motivated by the existence of contextual MBQC within proper subgroups of (tensor products of) $U(d)$, we ask whether (deterministic, non-adaptive) monomial $\zz_d$-MBQC in Def.~\ref{def: monomial Z_p-MBQC} gives rise to quantum solutions of its associated LCS. From now on, we will restrict to the case of odd prime (qudit) dimension, and write $d=p$.\\

Linear side-processing and (Clifford-restricted) $lp$-MBQC. To this end, recall first that Ref.~\cite{FrembsRobertsCampbellBartlett2022} proves that any function $o:\zz^l_p\ra\zz_p$ for $p$ prime can be computed within deterministic, non-adaptive $lp$-MBQC on a GHZ resource state using measurement operators restricted to the following unitary subgroup.\footnote{The linear side-processing restriction in $lp$-MBQC further takes inputs $I=\zz^l_p$ and pre-processing functions to be given by unitary conjugation, $M_k(\bi)=U(f_k(\bi))M_{k0}U^\dagger(f_k(\bi))$, with $f_k:\zz^l_p\ra\zz_p$ a linear function.\label{fn: lp-MBQC}} 

\begin{definition}\label{def: K-group}
    Let $X\in\cP(p)$ be the Pauli operator defined by $X|q\rangle=|q+1\mod p\rangle$ with $\{|q\rangle\}_{q\in\zz_p}$ a basis of $\cH=\C^p$ and let $S_\xi$ be the generalised phase gate defined by $S_\xi|q\rangle = \xi(q)|q\rangle$ where $\xi:\Z_p\to U(1)$.\footnote{We remark that our notation here slightly deviates from the one in Ref.~\cite{FrembsChungOkay2022}.\label{fn: xi convention}} Let\footnote{$T(p)$ is a maximal torus in the special unitary group. The latter naturally incorporates the $d$-torsion constraint (locally).}
    \begin{align*}
        T(p)
        =\{S_\xi\mid\xi:\zz_p\ra U(1),\ \det(S_\xi)=1\}\; ,
    \end{align*}
    and define the group
    \begin{align}\label{def: K-group}
        K(p)
        =\langle S_\xi X^b\mid b\in\zz_p, S_\xi\in T(p)\rangle<SU(p)\; .
    \end{align}
\end{definition}

In Ref.~\cite{FrembsChungOkay2022}, we further define the subgroups $K_Q(p)=\langle S_\xi X^b \mid b \in \zz_p, S_\xi\in Q \rangle<SU(p)$, where $T_{(p)}\leq Q\leq T_{(p^m)}$ for $T_{(p^m)}=\{S_\xi\in T(p)\suchthat S_\xi^{p^m}=\one\}$ the subgroups of elements with $p^m$-torsion and $m\in\N$. The groups $K_Q(p)$ extend the $n$-qudit Pauli group under the diagonal Clifford hierarchy \cite{CuiGottesmanKrishna2017,deSilva2021,FrembsChungOkay2022}.

Motivated by Thm.~\ref{thm: LCS = MBQC for Paulis}, and recalling the generic construction of (contextual) $lp$-MBQC in Ref.~\cite{FrembsRobertsCampbellBartlett2022} (which, in particular, are monomial $\zz_p$-MBQC), one may expect to obtain examples of LCS over $\zz_p$ with $p$ odd prime that admit quantum but no classical solutions, e.g. from LCS in Eq.~(\ref{eq: associated TG-LCS}) associated to contextual (deterministic, non-adaptive) $lp$-MBQC. In Ref.~\cite{FrembsChungOkay2022}, we obtain a partial negative result for this approach.

\begin{theorem}[Frembs-Chung-Okay \cite{FrembsChungOkay2022}]\label{thm: FCO}
    Let $\Gamma$ be the solution group of a LCS over $\zz_p$ for $p$ odd prime. Then the LCS admits a quantum solution $\eta: \Gamma \ra K^{\otimes n}_Q(p)$ if and only if it admits a classical solution $\eta:\Gamma\ra\zz_p$.
\end{theorem}

Thm.~\ref{thm: FCO} applies to contextual deterministic, non-adaptive $lp$-MBQC in Ref.~\cite{FrembsRobertsCampbellBartlett2022} as a special case. More precisely, state-dependent contextuality in $lp$-MBQC with operators in $K^{\otimes n}_Q(p)$ - which compute any function $o: \zz^l_p\ra\zz_p$ (with input group $I=\zz^l_p$) by Ref.~\cite{FrembsRobertsCampbellBartlett2022} - do not lift to state-independent contextuality (of their associated LCS).\footnote{However, there might be other (deterministic, non-adaptive) $lp$-MBQC, with measurement operators not constrained to $K^{\otimes n}_Q(p)$, that also compute arbitrary functions $o: \zz^l_p\ra\zz_p$. This further motives the generalisation to $\zz_p$-MBQC in Def.~\ref{def: Z_p-MBQC}.} Note, however, that neither the restriction (i) to phase gates in the Clifford hiearchy, nor (ii) to linear side-processing in MBQC, is necessary for the correspondence with LCS via Def.~\ref{def: associated LCS}.\\

\textbf{General side-processing in monomial $\zz_p$-MBQC.} Does Thm.~\ref{thm: FCO} generalise beyond the restriction to linear side-processing in $lp$-MBQC (with measurement operators given by unitary conjugation by elements in the Clifford hierarchy $K^{\otimes n}_Q(p)$) to monomial $\zz_p$-MBQC in Def.~\ref{def: monomial Z_p-MBQC}? More generally, do quantum solutions to (classically unsatisfiable) LCS exist within the subgroup $N_{SU(p)}(p)<SU(p)$? In other words, does Thm.~\ref{thm: FCO} still hold if $K_Q(p)$ is extended to $N_{SU(p)}(p)$? Our main results below will answer these questions affirmatively.

To get a sense for what this generalisation entails, note that $K(p)$ fits within the split exact sequence
\begin{align*}
    1 \ra T(p) \ra K(p) \ra \langle X\rangle\cong\zz_p \ra 1\; .
\end{align*}
It thus extends the Heisenberg-Weyl group, denoted by $H(p)$, by generalising the Pauli-$Z$ operator to more general phase gates $S_\xi$ (see Def.~\ref{def: K-group}), as is evident by writing $H(p)$ in terms of the split exact sequence
\begin{align}\label{ses: H}
    1\to \langle S_\xi\suchthat \mf{\xi(q)=\omega^{a+bq}},\  a,b\in\Z_p\rangle\to H(p)\to \langle X\rangle\cong\zz_p\to 1\; .
\end{align}
On the other hand, $K(p)$ is itself a subgroup of the normaliser $N_{SU(p)}(p)<SU(p)$ in the split exact sequence
\begin{equation}{\label{ses: N}}
    1 \ra T(p) \ra N_{SU(p)}(p) \ra S_p \ra 1\; ,
\end{equation}
which further generalises the shift permutation (Pauli-$X$) in Eq.~(\ref{def: K-group}) to general permutations $S_p$ on $p$ elements. As cited before, Ref.~\cite{QassimWallman2020} shows that every quantum solution to a LCS over $p$ odd restricted to the Heisenberg-Weyl group is classical. Our main technical contribution pushes this result to the normaliser subgroup.

\begin{theorem}\label{thm: FCO-2}
    Let $\Gamma$ be the solution group of a LCS over $\zz_p$ for $p$ odd prime. Then the LCS admits a quantum solution $\eta: \Gamma \ra N^{\otimes n}_{SU(p)}(p)$ if and only if it admits a classical solution $\eta:\Gamma\ra\zz_p$.
\end{theorem}

\begin{proof}[proof (sketch)]
    In a nutshell, this follows since commutation relations between elements in $N^{\otimes n}_{SU(p)}(p)$ are highly restricted: namely, two operators in $N^{\otimes n}_{SU(p)}(p)$ commute if and only if they lie in a subgroup conjugate to $H^{\otimes n}(p)$, and two such subgroups overlap only in $T^{\otimes n}_{(p)}$ (see Fig.~\ref{fig: H intersections}). For details, see App.~\ref{app: proof of main theorem}.
\end{proof}

Thm.~\ref{thm: FCO-2} may be of independent (mathematical) interest. As noted previously in Ref.~\cite{FrembsChungOkay2022}, a LCS with quantum solution in $G\leq\UH$ is classical if and only if there exists a homomorphism $\phi:G\ra\zz_p$ in $p$-torsion abelian subgroups. In fact, in the proof of Thm.~\ref{thm: FCO-2} we construct such a map for $G=N^{\otimes n}_{SU(p)}(p)$. The study of quantum solutions to linear constraint systems therefore corresponds, mathematically, with the study of homomorphisms in $p$-torsion abelian subgroups. Notably, both restrictions enlarge the respective sets of maps: clearly, a homomorphism of abelian subgroups generally does not extend to a full (group) homomorphism (cf. \cite{Semrl2008}); moreover, the restriction to $p$-torsion subgroups allows for further leeway, \mf{e.g. the map $\phi^N_H$ defined in the proof of Thm.~\ref{thm: FCO-2} (see App.~\ref{app: proof of main theorem}) cannot be lifted to a homormorphism in abelian subgroups (cf. App.~\ref{app: comparison}).}

As a special case, Thm.~\ref{thm: FCO-2} applies to the group $L(p)<N_{SU(p)}(p)$, defined via the split exact sequence\footnote{Recall from above that we write $G_{(p)}$ for the $p$-torsion elements in $G$.}
\begin{equation}\label{ses: L}
    1 \ra T_{(p)} \ra L(p) \ra S_p \ra 1\; .
\end{equation}
This group is a natural generalisation of the Heisenberg-Weyl group $H(p)$ in Eq.~(\ref{ses: H}) for $p$ odd, as well as of the qubit Pauli group for which $\langle X\rangle\cong S_2$, and further lies at the heart of the generalised stabiliser formalism in Ref.~\cite{VanDenNest2011,vanDenNest2013,BermejoVegaVanDenNest2014}. For these reasons, we analyse it in more detail in App.~\ref{app: comparison}.

Here, we focus on the immediate implication of Thm.~\ref{thm: FCO-2} for the question posed at the beginning.

\begin{theorem}\label{thm: no qsol from MBQC}
    A monomial $\zz_p$-MBQC $M$ with output function $o:I\ra\zz_p$ for $p$ odd prime constitutes a quantum solution to its associated LCS $A(M)x=b(o)\mod p$ in Def.~\ref{def: associated LCS} only if it is classical.
\end{theorem}

\begin{proof}
    Local measurement operators in monomial $\zz_p$-MBQC are elements in $N_{SU(p)}(p)<N_{U(p)}(p)$. A priori, operators may thus commute up to an arbitrary phase in local tensor factors. However, as we prove in Lm.~\ref{lm: qudit commutation relations reduce to Paulis} in App.~\ref{app: proof 2}, since for (global) measurement operators $p$-torsion holds locally, their (local) phase commutation relations are restricted to $p$-th roots of unity.  Consequently, the result follows from Thm.~\ref{thm: FCO-2}.
\end{proof}

Thm.~\ref{thm: no qsol from MBQC} proves the assertion in (\ref{eq: mismatch state-(in)dep contextuality}): unlike for qubits, tensor products of local measurement operators in $\zz_p$-MBQC (see Def.~\ref{def: Z_p-MBQC}) do not suffice to construct state-independent proofs of contextuality. It generalises the limited applicability of Thm.~\ref{thm: FCO} to $lp$-MBQC of the form in Ref.~\cite{FrembsRobertsCampbellBartlett2022}, by lifting the restrictions (i) on phase gates in the Clifford hierarchy and (ii) to linear side-processing. More generally, Thm.~\ref{thm: FCO-2} rules out quantum solutions in $N^{\otimes n}_{SU(p)}(p)$ to any classically unsatisfiable LCS, not only those of the form in Def.~\ref{def: associated LCS}.

\section{Conclusion}\label{sec: conclusion}

Taking the close correspondence between quantum solutions to (binary) linear constraint systems LCS (over $\zz_2$) and measurement-based quantum computation (MBQC) with qubits in the form of the Mermin-Peres square  \cite{Mermin1990,Peres1991} and Mermin's star \cite{Mermin1993} nonlocal games as our motivation, in this paper we analysed the question whether, more generally, state-dependently contextual MBQC with qudits of dimension $d>2$ gives rise to quantum solutions of a LCS naturally associated to it (see Def.~\ref{def: associated LCS}). We show that this is indeed the case under the restriction to the Pauli group (see Thm.~\ref{thm: LCS = MBQC for Paulis}), which for odd dimension, however, is known not to contain quantum solutions to LCS (over $\zz_d$) that are not also classically satisfiable \cite{QassimWallman2020}. We then lift the restriction to the Pauli group by considering more general monomial unitaries as measurement operators in MBQC (see Def.~\ref{def: monomial Z_p-MBQC}). This setting is both sufficiently general to comprise many instances of state-dependently contextual MBQC (Thm.~\ref{thm: AB qudit computation}, see also Ref.~\cite{FrembsRobertsCampbellBartlett2022}), and naturally embeds within the class of monomial matrices previously studied in Ref.~\cite{VanDenNest2011,vanDenNest2013}. Yet, even under this extension of the qudit Pauli group (and the resulting Clifford-type) hierarchy to normaliser subgroups of the special unitary group, our main result, Thm.~\ref{thm: no qsol from MBQC}, states that it, too, admits no quantum solutions to LCS over $\zz_d$ for $d$ odd prime that are not also classically solvable.

In summary, our work both generalises previous results \cite{QassimWallman2020,FrembsChungOkay2022}, and severely sharpens the known differences between qubit and qudit (state-independent) contextuality. While this leaves open the existence problem of LCS over $\zz_d$ for $d$ odd which admit quantum but no classical solutions, it highlights that to search for such examples one needs to study inherently state-independent proofs of quantum contextuality \cite{KochenSpecker1967}. A natural starting point in this direction is the algebraic (state-independent) reformulation of Kochen-Specker contextuality in Ref.~\cite{Frembs2024,Frembs2025}. \mf{We will discuss this together with other approaches in forthcoming work.}\\

\textbf{Acknowledgements.}
The second author acknowledges support from the Air Force Office of Scientific Research (AFOSR) under award number FA9550-24-1-0257 and the Digital Horizon Europe project FoQaCiA, GA no. 101070558.

\bibliographystyle{plain}
\bibliography{bibliography}

@article{howard2012qudit,
  title={Qudit versions of the qubit $\pi$/8 gate},
  author={Howard, Mark and Vala, Jiri},
  journal={Phys. Rev. A},
  volume={86},
  number={2},
  pages={022316},
  year={2012},
  publisher={APS}
}

@INPROCEEDINGS{CleveEtAl2004,
  author={Cleve, R. and Hoyer, P. and Toner, B. and Watrous, J.},
  booktitle={Proceedings. 19th IEEE Annual Conference on Computational Complexity, 2004.}, 
  title={Consequences and limits of nonlocal strategies}, 
  year={2004},
  volume={},
  number={},
  pages={236-249},
  doi={10.1109/CCC.2004.1313847}
}

@InProceedings{CleveMittal2014,
    author="Cleve, Richard
    and Mittal, Rajat",
    editor="Esparza, Javier
    and Fraigniaud, Pierre
    and Husfeldt, Thore
    and Koutsoupias, Elias",
    title="Characterization of Binary Constraint System Games",
    booktitle="Automata, Languages, and Programming",
    year="2014",
    publisher="Springer Berlin Heidelberg",
    address="Berlin, Heidelberg",
    pages="320--331"
}

@article{CleveLiuSlofstra2017,
    author = {Cleve, Richard and Liu, Li and Slofstra, William},
    title = {Perfect commuting-operator strategies for linear system games},
    journal = { J. Math. Phys.},
    volume = {58},
    number = {1},
    pages = {012202},
    year = {2017},
    doi = {10.1063/1.4973422}
}

@article{Slofstra2019,
   title={Tsirelson’s problem and an embedding theorem for groups arising from non-local games},
   volume={33},
   number={1},
   journal={J. Am. Math. Soc.},
   publisher={American Mathematical Society (AMS)},
   author={Slofstra, William},
   year={2019},
   month={Sep},
   pages={1–56}
}

@article{Slofstra2019b,
	title={The set of quantum correlations is not closed},
	volume={7},
	DOI={10.1017/fmp.2018.3},
	journal={Forum of Mathematics, Pi},
	publisher={Cambridge University Press},
	author={Slofstra, William},
	year={2019},
	pages={41}
}

@article{Mermin1990,
  title = {Simple unified form for the major no-hidden-variables theorems},
  author = {Mermin, N. David},
  journal = {Phys. Rev. Lett.},
  volume = {65},
  issue = {27},
  pages = {3373--3376},
  numpages = {0},
  year = {1990},
  month = {Dec},
  publisher = {American Physical Society}
}

@article{Mermin1993,
  title = "{Hidden variables and the two theorems of John Bell}",
  author = {Mermin, N. David},
  journal = {Rev. Mod. Phys.},
  volume = {65},
  issue = {3},
  pages = {803--815},
  numpages = {0},
  year = {1993},
  month = {Jul},
  publisher = {American Physical Society}
}

@article{Peres1991,
	year = 1991,
	month = {feb},
	volume = {24},
	number = {4},
	pages = {L175--L178},
	author = {Asher Peres},
	title = "{Two simple proofs of the Kochen-Specker theorem}",
	journal = {J. Phys. A}
}

@misc{JiEtAl2020,
    title={$\mathrm{MIP}^*=\mathrm{RE}$},
    author={Zhengfeng Ji and Anand Natarajan and Thomas Vidick and John Wright and Henry Yuen},
    year={2020},
    eprint={2001.04383},
    archivePrefix={arXiv},
    primaryClass={quant-ph}
}

@ARTICLE{HobanCampbell2011,
   author = {{Hoban}, Matty J and {Campbell}, Earl T and {Loukopoulos}, Klearchos and 
	{Browne}, Dan E},
    title = "{Non-adaptive measurement-based quantum computation and multi-party Bell inequalities}",
  journal = {New J. Phys.},
     year = 2011,
    month = feb,
   volume = 13,
   number = 2,
    pages = {023014},
      doi = {10.1088/1367-2630/13/2/023014},
}

@article{BravyiGossetKoenig2018,
	author = {Bravyi, Sergey and Gosset, David and K{\"o}nig, Robert},
	title = {Quantum advantage with shallow circuits},
	volume = {362},
	number = {6412},
	pages = {308--311},
	year = {2018},
	doi = {10.1126/science.aar3106},
	publisher = {American Association for the Advancement of Science},
	journal = {Science}
}

@article{BravyiGossetKoenig2020,
  title={Quantum advantage with noisy shallow circuits},
  author={Bravyi, Sergey and Gosset, David and K{\"o}nig, Robert and Tomamichel, Marco},
  journal={Nature Physics},
  volume={16},
  number={10},
  pages={1040--1045},
  year={2020},
  publisher={Nature Publishing Group}
}

@article{FrembsRobertsBartlett2018,
	doi = {10.1088/1367-2630/aae3ad},
	year = 2018,
	month = {Oct},
	publisher = {{IOP} Publishing},
	volume = {20},
	number = {10},
	pages = {103011},
	author = {Markus Frembs and Sam Roberts and Stephen D. Bartlett},
	title = {Contextuality as a resource for measurement-based quantum computation beyond qubits},
	journal = {New J. Phys.},
}

@article{FrembsChungOkay2022,
  doi = {10.22331/q-2025-01-08-1583},
  title = {No quantum solutions to linear constraint systems in odd dimension from {P}auli group and diagonal {C}liffords},
  author = {Frembs, Markus and Okay, Cihan and Chung, Ho Yiu},
  journal = {{Quantum}},
  issn = {2521-327X},
  publisher = {{Verein zur F{\"{o}}rderung des Open Access Publizierens in den Quantenwissenschaften}},
  volume = {9},
  pages = {1583},
  month = {Jan},
  year = {2025}
}

@ARTICLE{Gross2006,
   author = {{Gross}, David},
    title = "{Hudson's theorem for finite-dimensional quantum systems}",
  journal = { J. Math. Phys.},
   eprint = {quant-ph/0602001},
     year = 2006,
    month = dec,
   volume = 47,
   number = 12,
    pages = {122107-122107},
      doi = {10.1063/1.2393152}
}

@ARTICLE{AndersBrowne2009,
   author = {{Anders}, Janet and {Browne}, Dan E},
    title = "{Computational Power of Correlations}",
  journal = {Phys. Rev. Lett.},
     year = 2009,
    month = {Feb},
   volume = 102,
   number = 5,
    pages = {050502}
}

@article{deBeaudrap2013,
    author = {de Beaudrap, Niel},
    title = {A Linearized Stabilizer Formalism for Systems of Finite Dimension},
    journal = {Quantum Info. Comput.},
    issue_date = {January 2013},
    volume = {13},
    number = {1-2},
    month = jan,
    year = {2013},
    issn = {1533-7146},
    pages = {73--115},
    numpages = {43}
}

@ARTICLE{CuiGottesmanKrishna2017,
   author = {{Cui}, S.~X. and {Gottesman}, D. and {Krishna}, A.},
    title = "{Diagonal gates in the Clifford hierarchy}",
  journal = {Phys. Rev. Lett.},
     year = 2017,
    month = jan,
   volume = 95,
   number = 1,
      eid = {012329},
    pages = {012329},
      doi = {10.1103/PhysRevA.95.012329}
}

@article{BeigiShor2010,
  author = {Salman Beigi and Peter W Shor},
  title = {${C}_3$, semi-{C}lifford and genralized semi-{C}lifford operations},
  journal = {Quantum Inf. Comput.},
  volume = {10},
  number = {1{\&}2},
  pages = {41--59},
  year  = {2010},
  doi = {10.26421/QIC10.1-2-4},
}

@article{deSilva2021,
  title={Efficient quantum gate teleportation in higher dimensions},
  author={Nadish de Silva},
  journal={Proc. Math. Phys. Eng. Sci.},
  year={2021},
  volume={477}
}

@article{Raussendorf2016,
	author = {Robert Raussendorf},
	title = {Cohomological framework for contextual quantum computations},
	year = {2016},
	journal = {ArXiv e-prints},
	archivePrefix = "arXiv",
   	eprint = {1602.04155v2},
}

@ARTICLE{Raussendorf2013,
   author = {{Raussendorf}, Robert},
    title = "{Contextuality in measurement-based quantum computation}",
  journal = {Phys. Rev. A},
     year = 2013,
    month = aug,
   volume = 88,
   number = 2,
      eid = {022322},
    pages = {022322},
      doi = {10.1103/PhysRevA.88.022322}
}

@misc{Arkhipov2012,
      title={Extending and Characterizing Quantum Magic Games}, 
      author={Alex Arkhipov},
      year={2012},
      eprint={1209.3819},
      archivePrefix={arXiv},
      primaryClass={quant-ph}
}

@article{QassimWallman2020,
	doi = {10.1088/1751-8121/aba306},
	url = {https://doi.org/10.1088/1751-8121/aba306},
	year = 2020,
	month = {aug},
	publisher = {{IOP} Publishing},
	volume = {53},
	number = {38},
	pages = {385304},
	author = {Hammam Qassim and Joel J Wallman},
	title = {Classical vs quantum satisfiability in linear constraint systems modulo an integer},
	journal = {Journal of Physics A: Mathematical and Theoretical}
}

@article{JungeEtAl2011,
	title = "{Connes' embedding problem and Tsirelson's problem}",
	author = "Marius Junge and others",
	year = "2011",
	month = jan,
	day = "5",
	doi = "10.1063/1.3514538",
	volume = "52",
	journal = " J. Math. Phys.",
	issn = "0022-2488",
	publisher = "American Institute of Physics Publising LLC",
	number = "1",
}

@article{Fritz2012,
	author = {Fritz, Tobias},
	title = "{Tsirelson’s problem and Kirchberg’s conjecture}",
	journal = {Rev. Math. Phys.},
	volume = {24},
	number = {05},
	pages = {1250012},
	year = {2012},
	doi = {10.1142/S0129055X12500122}
}

@article{RaussendorfBriegel2001,
  title = {A One-Way Quantum Computer},
  author = {Raussendorf, Robert and Briegel, Hans J},
  journal = {Phys. Rev. Lett.},
  volume = {86},
  issue = {22},
  pages = {5188--5191},
  numpages = {0},
  year = {2001},
  month = {May},
  publisher = {American Physical Society},
  doi = {10.1103/PhysRevLett.86.5188}
}

@article{BriegelRaussendorf2001,
  title = {Persistent Entanglement in Arrays of Interacting Particles},
  author = {Briegel, Hans J. and Raussendorf, Robert},
  journal = {Phys. Rev. Lett.},
  volume = {86},
  issue = {5},
  pages = {910--913},
  numpages = {0},
  year = {2001},
  month = {Jan},
  publisher = {American Physical Society},
  doi = {10.1103/PhysRevLett.86.910}
}

@article{KochenSpecker1967,
	year = {1967},
	author = {Simon Kochen and E. P. Specker},
	title = {The Problem of Hidden Variables in Quantum Mechanics},
	pages = {59--87},
	journal = {J. Math. Mech.},
	volume = {17}
}

@article{FrembsRobertsCampbellBartlett2022,
  doi = {10.48550/ARXIV.2203.09965},
  author = {Frembs, Markus and Roberts, Sam and Campbell, Earl T. and Bartlett, Stephen D.},
  title = {Hierarchies of resources for measurement-based quantum computation},
  journal = {ArXiv e-prints},
  eprint={1607.05870},
  archivePrefix={arXiv},
  primaryClass={quant-ph},
  year = {2022}
}

@Article{Tsirelson1980,
	author="Tsirelson, Boris S.",
	title="Quantum generalizations of Bell's inequality",
	journal="Lett. Math. Phys.",
	year="1980",
	volume="4",
	number="2",
	pages="93--100",
	doi="10.1007/BF00417500"
}

@misc{Tsirelson2006,
    title={Bell inequalities and operator algebras}, 
    author={Boris S. Tsirelson},
    year={2006},
    howpublished = {\url{http://web.archive.org/web/20090414083019/http://www.imaph.tu-bs.de/qi/problems/33.html}}
}

@article{SlofstraPaddock,
    author = "Slofstra, William and Paddock, Connor",
    title = "{Satisfiability problems and algebras of Boolean constraint system games}",
    doi = "10.1215/00192082-11831196",
    journal = "Ill. J. Math.",
    volume = "69",
    number = "1",
    pages = "81--107",
    year = "2025"
}

@article{Semrl2008,
	title = {Commutativity preserving maps},
	journal = {Linear Alg. Appl.},
	volume = {429},
	number = {5},
	pages = {1051-1070},
	year = {2008},
	issn = {0024-3795},
	doi = {https://doi.org/10.1016/j.laa.2007.05.006},
	author = {Peter \v{S}emrl},
}

@article{TrandafirLisonekCabello2022,
  title = {Irreducible Magic Sets for $n$-Qubit Systems},
  author = {Trandafir, Stefan and Lison\ifmmode \check{e}\else \v{e}\fi{}k, Petr and Cabello, Ad\'an},
  journal = {Phys. Rev. Lett.},
  volume = {129},
  issue = {20},
  pages = {200401},
  numpages = {7},
  year = {2022},
  month = {Nov},
  publisher = {American Physical Society},
  doi = {10.1103/PhysRevLett.129.200401}
}

@article{MullerGiorgetti2025,
    author = {Muller, Axel and Giorgetti, Alain},
    title = {An abstract structure determines the contextuality degree of observable-based {K}ochen-{S}pecker proofs},
    journal = { J. Math. Phys.},
    volume = {66},
    number = {8},
    pages = {082203},
    year = {2025},
    month = {08},
    issn = {0022-2488},
    doi = {10.1063/5.0245341}
}

@Article{GottesmanChuang1999,
	author={Gottesman, Daniel and Chuang, Isaac L.},
	title={Demonstrating the viability of universal quantum computation using teleportation and single-qubit operations},
	journal={Nature},
	year={1999},
	month={Nov},
	volume={402},
	number={6760},
	pages={390-393},
	issn={1476-4687},
	doi={10.1038/46503}
}

@book{dummit2004abstract,
  title={Abstract algebra},
  author={Dummit, David Steven and Foote, Richard M and others},
  volume={3},
  year={2004},
  publisher={Wiley Hoboken}
}

@article{ZhouLeungChuang2000,
  title = {Methodology for quantum logic gate construction},
  author = {Zhou, Xinlan and Leung, Debbie W. and Chuang, Isaac L.},
  journal = {Phys. Rev. A},
  volume = {62},
  issue = {5},
  pages = {052316},
  numpages = {12},
  year = {2000},
  month = {Oct},
  publisher = {American Physical Society},
  doi = {10.1103/PhysRevA.62.052316}
}

@article{ZengChenChuang2008,
  title = {Semi-{C}lifford operations, structure of $C_k$ hierarchy, and gate complexity for fault-tolerant quantum computation},
  author = {Zeng, Bei and Chen, Xie and Chuang, Isaac L.},
  journal = {Phys. Rev. A},
  volume = {77},
  issue = {4},
  pages = {042313},
  numpages = {12},
  year = {2008},
  month = {Apr},
  publisher = {American Physical Society},
  doi = {10.1103/PhysRevA.77.042313}
}

@misc{ChenDeSilva2024,
      title={Characterising semi-{C}lifford gates using algebraic sets}, 
      author={Imin Chen and Nadish de Silva},
      year={2024},
      eprint={2309.15184},
      archivePrefix={arXiv},
      primaryClass={quant-ph}
}

@article{VanDenNest2011,
    doi = {10.1088/1367-2630/13/12/123004},
    year = {2011},
    month = {Dec},
    publisher = {IOP Publishing},
    volume = {13},
    number = {12},
    pages = {123004},
    author = {Nest, Maarten Van den},
    title = {A monomial matrix formalism to describe quantum many-body states},
    journal = {New J. Phys.}
}

@article{vanDenNest2013,
    author = {Van Den Nest, Maarten},
    title = {Efficient classical simulations of quantum fourier transforms and normalizer circuits over Abelian groups},
    year = {2013},
    publisher = {Rinton Press, Incorporated},
    address = {Paramus, NJ},
    volume = {13},
    number = {11–12},
    issn = {1533-7146},
    journal = {Quantum Info. Comput.},
    month = nov,
    pages = {1007–1037},
    numpages = {31},
}

@article{BermejoVegaVanDenNest2014,
    author = {Bermejo-Vega, Juan and Van Den Nest, Maarten},
    title = {Classical simulations of Abelian-group normalizer circuits with intermediate measurements},
    year = {2014}, 
    publisher = {Rinton Press, Incorporated},
    volume = {14},
    number = {3–4},
    issn = {1533-7146},
    journal = {Quantum Info. Comput.},
    month = {Mar},
    pages = {181–216},
    numpages = {36}
}

@article{AtseriasKolaitisSeverini2019,
    title = {Generalized satisfiability problems via operator assignments},
    journal = {Journal of Computer and System Sciences},
    volume = {105},
    pages = {171-198},
    year = {2019},
    issn = {0022-0000},
    doi = {https://doi.org/10.1016/j.jcss.2019.05.003},
    author = {Albert Atserias and Phokion G. Kolaitis and Simone Severini}
}

@InProceedings{BulatovStanislav2025,
  author =	{Bulatov, Andrei A. and \v{Z}ivn\'{y}, Stanislav},
  title =	{{Satisfiability of Commutative vs. Non-Commutative CSPs}},
  booktitle =	{52nd International Colloquium on Automata, Languages, and Programming (ICALP 2025)},
  pages =	{37:1--37:18},
  series =	{Leibniz International Proceedings in Informatics (LIPIcs)},
  ISBN =	{978-3-95977-372-0},
  ISSN =	{1868-8969},
  year =	{2025},
  volume =	{334},
  editor =	{Censor-Hillel, Keren and Grandoni, Fabrizio and Ouaknine, Jo\"{e}l and Puppis, Gabriele},
  publisher =	{Schloss Dagstuhl -- Leibniz-Zentrum f{\"u}r Informatik},
  address =	{Dagstuhl, Germany},
  doi =		{10.4230/LIPIcs.ICALP.2025.37}
}

@INPROCEEDINGS{Shamir1990,
  author={Shamir, Adi},
  booktitle={Proceedings [1990] 31st Annual Symposium on Foundations of Computer Science}, 
  title={{IP}={PSPACE} (interactive proof=polynomial space)}, 
  year={1990},
  volume={},
  number={},
  pages={11-15 vol.1},
  doi={10.1109/FSCS.1990.89519}
}

@Article{BabaiEtAl1991,
    author={Babai, L{\'a}szl{\'o} and Fortnow, Lance and Lund, Carsten},
    title={Non-deterministic exponential time has two-prover interactive protocols},
    journal={Computational complexity},
    year={1991},
    month={Mar},
    day={01},
    volume={1},
    number={1},
    pages={3-40},
    doi={10.1007/BF01200056}
}

@article{Bell1964,
    author = {Bell, John S},
    journal = {Physics},
    pages = {195},
    title = {On the {E}instein-{P}odolsky-{R}osen Paradox},
    volume = {1},
    year = {1964},
    month = {Nov},
    doi = {https://doi.org/10.1103/PhysicsPhysiqueFizika.1.195}
}

@article{Bell1975,
  title={The theory of local beables},
  author={Bell, John S},
  year={1975},
  month={Nov},
  journal={Epistemol. Lett.},
  volume={9},
  pages={11–24},
  doi={}
}

@article{FreedmanClauser1972,
  title = {Experimental Test of Local Hidden-Variable Theories},
  author = {Freedman, Stuart J and Clauser, John F},
  journal = {Phys. Rev. Lett.},
  volume = {28},
  issue = {14},
  pages = {938--941},
  numpages = {0},
  year = {1972},
  month = {Apr},
  publisher = {American Physical Society},
  doi = {10.1103/PhysRevLett.28.938}
}

@article{AspectEtAl1981,
  title = "{Experimental tests of realistic local theories via Bell's theorem}",
  author = {Aspect, Alain and Grangier, Philippe and Roger, G\'erard},
  journal = {Phys. Rev. Lett.},
  volume = {47},
  issue = {7},
  pages = {460--463},
  numpages = {0},
  year = {1981},
  month = {Aug},
  publisher = {American Physical Society},
  doi = {10.1103/PhysRevLett.47.460}
}

@article{ZeilingerEtAl2015,
  title = "{Significant-loophole-free test of Bell's theorem with entangled photons}",
  author = {Giustina, Marissa and others},
  journal = {Phys. Rev. Lett.},
  volume = {115},
  issue = {25},
  pages = {250401},
  numpages = {7},
  year = {2015},
  month = {Dec},
  publisher = {American Physical Society},
  doi = {10.1103/PhysRevLett.115.250401}
}

@article{ShalmEtAl2015,
  title = {Strong Loophole-Free Test of Local Realism},
  author = {Shalm, Lynden K. and others},
  journal = {Phys. Rev. Lett.},
  volume = {115},
  issue = {25},
  pages = {250402},
  numpages = {10},
  year = {2015},
  month = {Dec},
  publisher = {American Physical Society},
  doi = {10.1103/PhysRevLett.115.250402}
}

@misc{Frembs2024,
      title={An algebraic characterisation of {K}ochen-{S}pecker contextuality}, 
      author={Markus Frembs},
      year={2024},
      eprint={2408.16764},
      archivePrefix={arXiv},
      primaryClass={quant-ph}
}

@misc{Frembs2025,
      title="{Coming full circle -- A unified framework for Kochen-Specker contextuality}", 
      author={Markus Frembs},
      year={2025},
      eprint={2501.09750},
      archivePrefix={arXiv},
      primaryClass={quant-ph}
}

\appendix

\section{Proof of Thm.~\ref{thm: FCO-2}}\label{app: proof of main theorem}

The generalisation of Thm.~\ref{thm: FCO} in Thm.~\ref{thm: FCO-2} involves two steps: (i) allowing for arbitrary subgroups $Q\leq T(p)$ (generalised phase gates, not restricted to the Clifford hierarchy), and (ii) extending the abelian group generated by shifts (Pauli-X) to the full permutation group $S_p$. Both steps essentially follow from Lm.~\ref{lem: comm_tildeK} below, which generalises Lm.~4 in Ref.~\cite{FrembsChungOkay2022}. We will further discuss the relation with our previous results in App.~\ref{app: comparison}.\\

\textbf{Notation.} Let $T=T(p)$, $H=H(p)$, $K=K(p)$ and $N=N_{SU(p)}(p)$. We denote elements in $N$ by $(\xi,\sigma)$ with $\xi\in T$ and $\sigma\in S_p$, and often use their explicit matrix representation, where $\xi\in T$ corresponds to the diagonal operator $S_\xi=\mathrm{diag}\mf{(\xi(0),\cdots,\xi(p-1))}$ and $\sigma\in S_p$ to the permutation matrix acting by conjugation, that is, by
$$\sigma\cdot S_\xi=\begin{pmatrix}\xi\circ \sigma(0) & &0\\ & \ddots & \\ 0& & \xi\circ\sigma(p-1)\end{pmatrix}\; .$$
More generally, for $M=(\xi,\sigma)\in N$ and $\tau\in S_p$, we will write
$\tau\cdot M=(0,\tau)(\xi,\sigma)(0,\tau^{-1})=(\tau\cdot\xi,\tau\sigma\tau^{-1})$, hence, the group operation is given by $(\xi,\sigma)(\chi,\tau)=(\xi+\sigma\cdot\chi,\sigma\tau)$. We denote the identity element in $S_p$ by $\id$.

\begin{obv}{\label{obv_comm}}
    Let $M,M'\in N$ and $\tau\in S_p$, then $[\tau\cdot M,\tau\cdot M']=\tau\cdot[M,M']$.
\end{obv}

\begin{lemma}{\label{lem: sp generator}}
    Let $G\leq S_p$, where $G\cong \Z_p$. Then we can pick $g=(0\,1\,a_2\,\cdots\,a_{p-1})$, where $a_i\neq a_j$ and $a_i\in\{2,...,p-1\}$ to be the generator of $G$. In addition, let $\tau=(a_2\,\,2)(a_3\,\,3)\cdots(a_{p-1}\,\,p-1)$, then we have $\tau g \tau^{-1}=(0\,1\,\cdots\,p-1)$.
\end{lemma}

\begin{proof}
    Since $G$ is a Sylow $p$ subgroup of $S_p$, by \cite[Sylow's Theorem]{dummit2004abstract}, there exists $\sigma\in S_p$ such that $G=\sigma \langle X\rangle\sigma^{-1}$. By \cite[Proposition 10]{dummit2004abstract}, %page 125
    we have $G=\langle (\sigma(0)\,\,\cdots\,\,\sigma(p-1))\rangle$. Observe that there exists $k\in\{1,...,p-1\}$ such that $(\sigma(0)\cdots\sigma(p-1))^k=(0\,\,1\,\,a_2\cdots \,\,a_{p-1})$, where $a_i\in\{2,...,p-1\}$. Thus, we have $G=\langle (0\,1\,a_2\,\cdots\,a_{p-1})\rangle$. Let $\tau=(a_2\,\,2)(a_3\,\,3)\cdots(a_{p-1}\,\,p-1)$, using \cite[Proposition 10]{dummit2004abstract} again, %page 125,
    we find $\tau (0\,1\,a_2\cdots a_{p-1})\tau^{-1}=(0\,1\,\cdots\,p-1)$.
\end{proof}

Note that there are $(p-1)!$ distinct elements of order $p$ (`$p$-cycles') in $S_p$. Since every $\zz_p$-subgroup (isomorphic to $\Z_p$) in $S_p$ contains $p-1$ elements, there are $\frac{(p-1)!}{p-1}=(p-2)!$ distinct $\Z_p$-subgroups in $S_p$. With Lm.~\ref{lem: sp generator}, we can associate to every one of these conjugate subgroups a unique conjugating permutation.

\begin{definition}{\label{defn: tildeK_sp}}
    Let $G_1,...,G_{(p-2)!}$ denote the distinct $\zz_p$-subgroups of $S_p$, that is, $G_i\cong \Z_p$. With Lm.~\ref{lem: sp generator}, fix generators $g_i$ of $G_i$ by $g_1=(0\,1\,...\,p-1)$ and $g_i=(0\,1\,a^{(i)}_2\,...\,a^{(i)}_{p-1})$ where $a^{(i)}_j\neq a^{(i)}_k$ and $a^{(i)}_k\in\{2,...,p-1\}$ for all $i\in \{2,...,(p-2)!\}$ and $j,k\in\{2,...,p-1\}$. Moreover, let $\tau_i=(a^{(i)}_3\,\,3)(a^{(i)}_4\,\,4)\cdots(a^{(i)}_p\,\,p)$ such that $\tau_i g_i \tau^{-1}_i=g_1$.
    
    Let $proj: N\to S_p$ be the projection map in Eq.~(\ref{ses: N}), and identify $(0\,\,1\,\,2\cdots\,\,p-1)$ with $proj(X)$, where $X$ is the Pauli $X$ operator. Then define the map $\phi^N_{S_p}:N\to S_p$ by
    $$
    \phi^N_{S_p}(M) = \begin{cases}
      \tau_i & \text{if $proj(M)\in G_i$ and $proj(M)\neq\id$} \\
      \id &  \text{otherwise}
    \end{cases}\; .
    $$
\end{definition}

Using the map $\phi^N_{S_p}$, we can reduce commutation relations in $N$ to those in $K$.

\begin{obv}{\label{obv:tildeK to K}}
    \begin{enumerate}
        \item Let $M\in N$, then $\phi^N_{S_p}(M)\cdot M\in K$
        \item Let $M,M'\in N$ such that $M\not\in T$ and $M'\in T$, then $\phi^N_{S_p}(MM')=\phi^N_{S_p}(M)$.
    \end{enumerate}
\end{obv}

With this, we study commutation relations up to phase in $N$, generalising our previous Lm.~4 in Ref.~\cite{FrembsChungOkay2022}. 

\begin{lemma}\label{lem: comm_tildeK}
    Let $M=(\xi,\sigma),M'=(\xi',\sigma')\in N$ with $M^p,M'^p=\one$ and $[M,M']\in \langle\omega\one\rangle=Z(N)$ for $\omega=e^{\frac{2\pi i}{p}}$.
    \begin{enumerate}
    \item If $\sigma=\sigma'=\id$, then $M,M'\in T$.
    \item If $\sigma=\id$ and $\sigma'\neq\id$, then $\phi^N_{S_p}(M')\cdot M\in H\cap T_{(p)}$. In particular, if $[M,M']=\one$, then $M\in Z(N)$.
    \item If $\sigma,\sigma'\neq\id$ then $\phi^N_{S_p}(M)=\phi^N_{S_p}(M')$ and there exists $a\in\Z_p$ and $(\chi,\id)\in H\cap T_{(p)}$ such that
    $$\left(\phi^N_{S_p}(M')\cdot M'\right)=\left((\phi^N_{S_p}(M)\cdot M)(\chi,\id)\right)^a$$ In particular, if $[M,M']=\one$, then $(\chi,\id)\in Z(N)$.
    \end{enumerate}
\end{lemma}
\begin{proof}
    Case 1 is trivial.
    
    For case 2, let $\tau=\phi^N_{S_p}(M')$. By Observation \ref{obv:tildeK to K}, $\tau\cdot M\in T$ and $\tau\cdot M'\in K$. On the other hand, we have
    $$[\tau\cdot M,\tau\cdot M']=\tau\cdot[M,M']\in Z(K)=Z(N)$$
    Hence, by \cite[Lm.~4]{FrembsChungOkay2022},\footnote{While not stated explicitly, \cite[Lm.~4]{FrembsChungOkay2022} is easily seen to apply to $K$, not only to $K_Q$.} $\tau\cdot M\in H\cap T_{(p)}=\langle\omega\one\rangle\times\hat{\Z}_p$, \mf{where $\hat{\Z}_p=\{S_\xi\in T(p)\suchthat \xi(q)=\omega^{aq}, a\in \Z_p\}<T_{(p)}$ denotes the Pontryagin dual of $\zz_p$}. In particular, if $[M,M']=\one$, then $[\tau\cdot M,\tau\cdot M']=\one$, hence, \cite[Lm.~4]{FrembsChungOkay2022} asserts that $\tau\cdot M\in Z(N)$ and thus $M\in Z(N)$. (Analogously, for $\sigma\neq\id$, $\sigma'=\id$.)

    For case 3, since $[M,M']\in Z(N)$, we have $\sigma\sigma'=\sigma'\sigma$. The subgroup $\langle \sigma,\sigma'\rangle\leq S_p$ is thus isomorphic either to $\Z_p$ or $\Z_p\times \Z_p$. Yet, the latter is impossible, since $S_p$ (of order $p!$) contains no subgroup of order $p^2$. We conclude that $\langle\sigma,\sigma'\rangle\cong\Z_p$ which implies $\sigma'=\sigma^a$ for some $a\in\zz_p$ and thus $\phi^N_{S_p}(M)=\phi^N_{S_p}(M')=:\tau$. By Observation \ref{obv_comm}, $[\tau\cdot M,\tau\cdot M']=\tau\cdot[M,M']\in Z(N)$, hence, the result follows again from \cite[Lemma 4]{FrembsChungOkay2022}.
\end{proof}

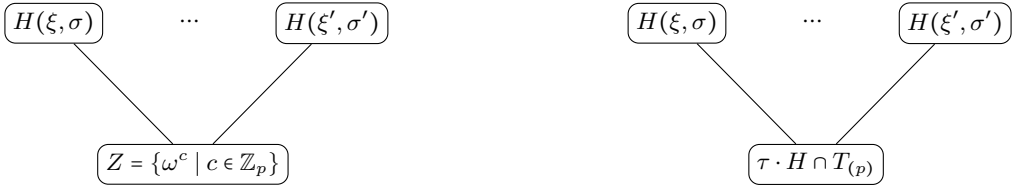
\begin{figure}[t]
\centering

\begin{tikzpicture}[
  every node/.style={font=\small},
  grp/.style={draw, rounded corners, inner sep=3pt, align=center},
  dots/.style={inner sep=0pt},
  arr/.style={->, line width=0.4pt}
]

% ---------------- Left subfigure ----------------
\begin{scope}
  \node[grp] (H1) {$H(\xi,\sigma)$};
  \node[dots, right=10mm of H1] (hdots) {$\cdots$};
  \node[grp, right=10mm of hdots] (H2) {$H(\xi',\sigma')$};

  \node[grp, below=16mm of $(H1)!0.5!(H2)$] (Z)
    {$Z=\{\omega^c\mid c\in\mathbb{Z}_p\}$};

  \draw (Z) -- (H1);
  \draw (Z) -- (H2);
\end{scope}

% ---------------- Right subfigure ----------------
\begin{scope}[xshift=8.2cm]
  \node[grp] (K1) {$H(\xi,\sigma)$};
  \node[dots, right=10mm of K1] (kdots) {$\cdots$};
  \node[grp, right=10mm of kdots] (K2) {$H(\xi',\sigma')$};

  \node[grp, below=16mm of $(K1)!0.5!(K2)$] (Tp)
    {$\tau\cdot H\cap T_{(p)}$};

  \draw (Tp) -- (K1);
  \draw (Tp) -- (K2);
\end{scope}

\end{tikzpicture}

\caption{Two types of intersections between conjugate Heisenberg-Weyl subgroups $H\cong H(\xi,\sigma),H(\xi',\sigma')\leq N$ (see Lm.~\ref{lem: comm_tildeK}), where $\xi,\xi'\in T$ and $\sigma,\sigma'\in S_p$ (with notation as in proof of Thm.~\ref{thm: FCO-2}). If $(\xi',\sigma')\neq(\xi\chi,\sigma)^a$ for all $a\in\zz_p$ and  $(\chi,\id)\in\tau\cdot H\cap T_{(p)}\cong\{\omega^cZ^b\mid b,c\in\zz_p\}$ (with $Z$ the Pauli-Z operator in $H$) and $\tau=\phi^N_{S_p}(\xi,\sigma)$, the intersection is the center (left); if instead $(\xi',\sigma')\neq(\xi\chi,\sigma)^a$ (for some $a\in\zz_p$ and $\chi\in\tau\cdot H\cap T_{(p)}$ with $\tau=\phi^N_{S_p}(\xi,\sigma)$), the intersection additionally contains the subgroup $\tau\cdot\langle Z\rangle$, conjugated by $\tau$ (right).}
\label{fig: H intersections}
\end{figure}

Lm.~\ref{lem: comm_tildeK} says that commutation relations up to phase between elements in $N$ restrict to the various conjugate subgroups of $H$ in $N$, which are of the form $H\cong H(\xi,\sigma):=\langle(\xi,\sigma),(\chi,\id)\rangle$ with $(\chi,\id)\in T_{(p)}$ such that $\tau\cdot(\chi,\id)\in H$ for $\tau=\phi^N_{S_p}(\xi,\sigma)$ as defined in Def.~\ref{defn: tildeK_sp}.\footnote{This essentially decomposes elements in $N$ into diagonal elements in $T_{(p)}$ and `generalised permutations' of the form $(\xi,\sigma)$.} Here, $H$ corresponds with $\xi=1$, that is, $S_\xi=\diag(1,\cdots,1)$ and $\sigma=(0\cdots p-1)=proj(X)$ (represented by the Pauli-X operator). Moreover, Lm.~\ref{lem: comm_tildeK} shows that the intersection between different $H$-subgroups is highly restricted, to subgroups of $T_{(p)}$, as depicted in Fig.~\ref{fig: H intersections}.

\begin{proof}[Proof of Thm.~\ref{thm: FCO-2}]
    By Thm.~2 in Ref.~\cite{QassimWallman2020}, there exists a map $v_{H^{\otimes n}}:H^{\otimes n}\ra\zz_p$ that is a homomorphism in abelian subgroups.\footnote{Note that every element in $H^{\otimes n}$ has $p$-torsion.} We will extend $v_{H^{\otimes n}}$ to a homomorphism in $p$-torsion abelian subgroups $v_{N^{\otimes n}}:N^{\otimes n}\ra\zz_p$. 
    
    To this end, note that by Lm.~\ref{lem: comm_tildeK}, the $p$-torsion abelian subgroups in $N$ are, first, $T_{(p)}<T$ and, second, (non-diagonal) subgroups of the form $A(\xi,\sigma)=\langle(\xi,\sigma)\rangle$, where $\sigma\neq\id$, $\sigma^p=\id$ and $\xi\in T$. These are contained in the respective isomorphic copies $H\cong H(\xi,\sigma)<N$, and we have
    \begin{align}\label{eq: H-subgroup decomposition}
        N_{(p)}
        =\bigcup_{\substack{\id\neq\sigma\in S_p\\ \sigma^p=\id,\xi\in T}} H(\xi,\sigma)
        =\bigcup_{\xi\in T_{(p)}}A(\xi,\id)\cup\bigcup_{\substack{\id\neq\sigma\in S_p\\ \sigma^p=\id,
        \xi\in T}} A(\xi,\sigma)\; .
    \end{align}
    By Lm.~\ref{lem: comm_tildeK}, $H(\xi,\sigma)=H(\xi',\sigma')$ if and only if $(\xi',\sigma')=(\xi,\sigma)^a(\chi,\id)$ for $a\in\zz_p$ (and $(\chi,\id)\in T_{(p)}$ with $\tau\cdot(\chi,\id)\in H$ as above, see also Fig.~\ref{fig: H intersections}). For any $p$-torsion abelian subgroup $H\not>A(\xi,\sigma)<N$ (with either $\sigma\neq\id$, $\sigma^p=\id$ or $A(\xi,\sigma)=A(\chi,\id)$ for $\chi\in T_{(p)}$ by Eq.~(\ref{eq: H-subgroup decomposition})), we may thus pick a generator $(\xi,\sigma)$, and define\footnote{We will use the shorthand $v_N(\xi,\sigma):=v_N((\xi,\sigma))$ to avoid clutter.}
    \begin{align}\label{eq: def v_N}
        v_N(\xi,\sigma)\
        :=\ \begin{cases}
            \ v_H(\phi^{N(p)}_{S_p}(\xi,\sigma)\cdot(\xi,\sigma))  & \quad\mathrm{if}\ \sigma\neq\id,\sigma^p=\id\ \mathrm{and}\ \phi^{N(p)}_{S_p}(\xi,\sigma)\cdot(\xi,\id)\in H \\
            \ v_H(\phi^{N(p)}_{S_p}(\xi,\sigma)\cdot(1,\sigma)) & \quad\mathrm{if}\ \sigma\neq\id,\sigma^p=\id\ \mathrm{and}\ \phi^{N(p)}_{S_p}(\xi,\sigma)\cdot(\xi,\id)\notin H \\
            \ v_H(\tau\cdot(\xi,\id)) & \quad\mathrm{if}\ \sigma=\id,(\xi,\id)\in T_{(p)}\ \mathrm{and}\ \exists\ \tau\in\mathrm{Im}(\phi^N_K)\ \mathrm{s.t.}\ \tau\cdot(\xi,\id)\in H \\
            \ 1 & \quad\mathrm{if}\ \sigma=\id,(\xi,\id)\in T_{(p)}\ \mathrm{and}\ \not\exists\tau\in\mathrm{Im}(\phi^N_K)\ \mathrm{s.t.}\ \tau\cdot(\xi,\id)\in H
        \end{cases}
    \end{align}
    Note that $\tau$ is necessarily unique if its exists, as follows from Def.~\ref{defn: tildeK_sp} and Lm.~\ref{lem: comm_tildeK} (2). Next, we extend $v_N$ multiplicatively to $A(\xi,\sigma)$, simply by demanding $v_N((\xi,\sigma)^a):=(v_N(\xi,\sigma))^a$ for all $a\in\zz_p$. Clearly, this is consistent with $v_N|_H=v_H$, in particular, with $v_N|_{\hat{\zz}_p}=v_H|_{\hat{\zz}_p}$ in $H(\xi,\sigma)$, \mf{moreover with $T_{(p)}$, since elements that commute up to phase by Lm.~\ref{lem: comm_tildeK} are always conjugated by the same permutation.}
    By Eq.~(\ref{eq: H-subgroup decomposition}), defines $v_N$ on all of $N_{(p)}$. Finally, since, by Lm.~\ref{lem: comm_tildeK}, all commutation relations up to phase arise in subgroups $H(\xi,\sigma)$, $v_{N^{\otimes n}}:N^{\otimes n}_{(p)}\ra\zz_p$ defined by $v_{N^{\otimes n}}:=\prod_{k=1}^n v_{N_k}$ is a homomorphism in $p$-torsion abelian subgroups as desired.
\end{proof}

$v_N$ simply copies $v_H$ to every isomorphic copy $H(\xi,\sigma)$ of $H$ in $N$. This mapping is rather arbitrary, especially the second case of Eq.~(\ref{eq: def v_N}) can be defined in various inequivalent ways without spoiling the result. Notably, $v_N$ is thus more loosely defined than our previous map $\phi:(K^{\otimes n}_Q)_{(p)}\ra H^{\otimes n}$ in Ref.~\cite{FrembsChungOkay2022}. In App.~\ref{app: comparison}, we further compare these two constructions and discuss to what extent the properties of $\phi$ carry over to $N$.
\section{Commutation up to phase under local $p$-torsion}\label{app: proof 2}

The measurement operators in $\zz_p$-MBQC are $p$-torsion unitaries in every tensor factor. If two such (global) operators commute, this implies that their (local) tensor factors commutate up to a $p$-th root of unity.

\begin{lemma}\label{lm: qudit commutation relations reduce to Paulis}
    Let $u=\otimes_{k=1}^n u_k, u'=\otimes_{k=1}^n u'_k \in\UH^{\otimes n}$ such that $u_k^d=u'^d_k=\one$ for all $k\in\{1,\cdots,n\}$ and $d$ odd.\footnote{To avoid clutter, we suppress the site label for identity elements, writing $\one_k=\one$ for all $k\in[n]$, and similarly for subsets $I\subset[n]$.} Then $[u,u']=0$ if and only if $u_ku'_k=\omega^{l_k}u'_ku_k$ with $l_k\in\zz_d$ for all $k\in\{1,\cdots,n\}$ and such that $\sum_k l_k=0\mod d$.
\end{lemma}

\begin{proof}
    We first consider the bipartite case ($n=2$), and decompose the commutator of tensor products as tensor products of (anti-)commutators (here, $[u,u']=\frac{1}{2}(uu'-u'u)$ and $\{u,u'\}=\frac{1}{2}(uu'+u'u)$) via the formula,
    \begin{align}\label{eq: commutator factorisation}
        [u_1\otimes u_2,u'_1\otimes u'_2]=[u_1,u'_1]\otimes\{u_2,u'_2\}+\{u_1,u'_1\}\otimes[u_2,u'_2]\; .
    \end{align}
    Clearly, $[u_1,u'_1]=0=[u_2,u'_2]$ (with $l_1=0=l_2$) implies $[u_1\otimes u_2,u'_1\otimes u'_2]=0$; moreover, (omitting indices) the case $\{u,u'\}=0$ is excluded, since this is equivalent to $u'uu'^{-1}=-u$ and thus $\mathbbm{1}=u'u^du'^{-1}=(u'uu'^{-1})^d=-u^d=-\mathbbm{1}$, where we used that $u^d=\mathbbm{1}$ and that $d$ is odd. We may thus assume $[u_1,u'_1],\{u_1,u'_1\}\neq 0\neq [u_2,u'_2],\{u_2,u'_2\}$, which leaves the case $[u_1,u'_1]\otimes\{u_2,u'_2\}=-\{u_1,u'_1\}\otimes[u_2,u'_2]$, that is, $[u_k,u'_k]=\alpha_k\{u_k,u'_k\}$ with $k=1,2$ and $\alpha_1=-\alpha_2\in U(1)$. It follows that (omitting indices) $uu'=[u,u']+\{u,u'\}=(\alpha+1)\frac{1}{2}(uu'+u'u)$, which is only possible if $uu'\propto u'u$ (since $0\notin\UH$ and thus $u'u\neq 0$). Let thus $u'u=\lambda uu'$, then $uu'=\frac{1}{2}(\alpha+1)(\lambda+1)uu'$, hence, $\lambda=\frac{2}{\alpha+1}-1$. But since $\mathbbm{1}=u'u^du'^{-1}=(u'uu'^{-1})^d=\lambda^d u^d=\lambda^d\mathbbm{1}$, this implies $\lambda^d=(\frac{2}{\alpha+1}-1)^d=1$. We may thus set $\lambda_k=\frac{2}{\alpha_k+1}-1=e^{\frac{2\pi i}{d}l_k}=\omega^{l_k}$, that is, $\alpha_k=\frac{1-\omega^{l_k}}{1+\omega^{l_k}}$ for some $l_k\in\zz_d$ and $k=1,2$. With $\alpha_1=-\alpha_2$, we thus find $-(1-\omega^{l_1})(1+\omega^{l_2})=(1-\omega^{l_2})(1+\omega^{l_1})$ which simplifies to $1=\omega^{l_1+l_2}$, hence, $l_1+l_2=0\mod d$.

    Finally, since $u^d_k=\one$ for all $k\in[n]$ also implies that $(\otimes_{k\in I\subset[n]}u_k)^d=\one$, we can apply the above argument to every bi-partition. The result thus follows by decomposing the $n$-fold tensor product into local factors.

    The converse follows immediately for $l_1=0=l_2$, and for any pair $l_k,l'_k\neq 0$ since $u_ku'_k=\omega^{l_k}u'_ku_k$ implies both $u'u\propto uu'$ and $[u,u']\propto\{u,u'\}$ with the proportionality constants fixed by $l_k,l'_k$ as above.
\end{proof}

If we relax the assumption on $d$-torsion in Lm.~\ref{lm: qudit commutation relations reduce to Paulis} \mf{to $u^d=u'^{d'}=\mathbbm{1}$ for $d,d'$ odd}, then $[u,u']=0$ if and only if $uu'=\omega^lu'u$ for $l\in\zz_{g}$, $\omega=e^{\frac{2\pi i}{g}}$ and $g=\mathrm{gcd}(d,d')$. In particular, if $\mathrm{gcd}(d,d')=1$ then necessarily $[u,u']=0$.

Lm.~\ref{lm: qudit commutation relations reduce to Paulis} can further be adapted such as to apply to systems of even dimension. We consider the case $d=2$.

\begin{lemma}\label{lm: qubit commutation relations reduce to Paulis}
    Let $u=\otimes_{k=1}^n u_k, u'=\otimes_{k=1}^n u'_k \in\UH^{\otimes n}$ such that $u_k^2=u'^2_k=\mathbbm{1}$ for all $k\in\{1,\cdots,n\}$. Then $[u,u']=0$ if and only if $u_ku'_k=(-1)^{l_k}u'_ku_k$ with $l_k\in\zz_2$ for all $k\in\{1,\cdots,n\}$ and such that $\sum_k l_k=0\mod 2$.
\end{lemma}

\begin{proof}
    \mf{Following the proof of Lm.~\ref{lm: qudit commutation relations reduce to Paulis}}, we find $\lambda_k=\pm 1$ for $k=1,2$, which corresponds with $[u_1,u'_1]=0$ and $\{u_1,u'_1\}=0$, respectively. In the first case, $\{u_1,u'_1\}\neq 0$, hence, $[u_2,u'_2]=0$, in the second, $[u_2,u'_2]\neq 0$, hence, $\{u_2,u'_2\}=0$. It follows that $\lambda_1=\lambda_2$, equivalently, $1=\lambda_1\lambda_2=(-1)^{l_1+l_2}$, hence, $l_1+l_2=0\mod 2$ for $l_1,l_2\in\zz_2$.

    Finally, since unlike in the case for $d$ odd in Lm.~\ref{lm: qudit commutation relations reduce to Paulis}, $\{u,u'\}=0$ is a possibility for $l=1$ ($\lambda=-1$), when considering consecutive bipartitions, we thus also need to consider decompositions of anti-commutators
    \begin{align*}
        \{u_1\otimes u_2,u'_1\otimes u'_2\}=[u_1,u'_1]\otimes[u_2,u'_2]+\{u_1,u'_1\}\otimes\{u_2,u'_2\}\; .
    \end{align*}
    Again, we have $\lambda_k=\pm 1$ and by similar reasoning as above, $\lambda_1=-\lambda_2$, equivalently $-1=(-1)^{l_1+l_2}$, that is, $1=l=l_1+l_2\mod 2$. Consequently, upon fully decomposing the commutator, we find $\sum_k l_k=0\mod 2$.
\end{proof}

Note that Lm.~\ref{lm: qubit commutation relations reduce to Paulis} recovers the well-known commutation relations between qubit Pauli operators.
\section{Comparison with previous results}\label{app: comparison}

In Ref.~\cite{FrembsChungOkay2022}, we construct a homomorphism of abelian $p$-torsion subgroups $\phi:K^{\otimes n}_Q\ra H^{\otimes n}$ with $T_{(p)}\leq Q\leq T_{(p^m)}$ and $m\in\mathbb{N}$. As emphasised in the main part, Thm.~\ref{thm: FCO-2} generalises this in two directions: (i) by lifting the restriction on subgroups $Q\leq T$ and (ii) by allowing for general permutations $S_p$. Note, however, that the proof of Thm.~\ref{thm: FCO-2} in App.~\ref{app: proof of main theorem} does not generalise the map $\phi$ directly. In this section, we apply our previous results to obtain a restricted version of Thm.~\ref{thm: FCO-2} for the extension of the groups $K_Q$ to permutations in
\begin{align*}
    1 \ra Q \ra N_Q \ra S_p \ra 1\; .
\end{align*}
To this end, we focus on the group $L$ in Eq.~(\ref{ses: L}), which forms a natural generalisation of the Heisenberg-Weyl and qubit Pauli groups, and may thus be of independent interest. For convenience, we write $K=K_{T_{(p)}}(p)$.\footnote{This should not be confused with our shorthand for the group $K_Q$ in Ref.~\cite{FrembsChungOkay2022}.}

\subsection{Extending the Heisenberg-Weyl group by arbitrary permutations}

In this section, we construct a homomorphism in abelian, $p$-torsion subgroups $\phi^{\otimes L}_{\otimes H}:L^{\otimes n}\ra H^{\otimes n}$ for the group defined in Eq.~(\ref{ses: L}), which extends the Heisenberg-Weyl group by taking into account not only shifts, but arbitrary permutations in $S_p$. To this end, we will define an explicit map that conjugates the various abelian subgroups in $L$ to ones in $H$. On non-diagonal elements this is simply the restriction of the map $\phi^N_H|_L$, defined in Def.~\ref{defn: tildeK_sp}. We thus need to extend it to the elements in $T_{(p)}$. Recall from \cite[Theorem 3]{FrembsChungOkay2022} that $M_\xi\in T_{(p)}$ if and only if \mf{$\xi(q)=\omega^{\sum_{j=0}^{p-1} a_jq^j}$} for $a_j\in\Z_p$.$^{\ref{fn: xi convention}}$ We first define the following maps.

\begin{definition}\label{def: phi^K_H}
    Define $\phi^{T_{(p)}}_H:T_{(p)}\to H$ by $M_\xi\mapsto M_\chi$, \mf{where $\xi(q)=\omega^{\sum_{j=0}^{p-1} a_jq^j}$ and $\chi(q)=\omega^{a_0+(\sum_{j=1}^{p-1}a_j)q}$}. Moreover, define $\phi^K_H:K\to H$ by its action on elements $M=S_\xi X^b\in K$, with $\xi\in T_{(p)}$ and $X$ the permutation matrix represented by the Pauli $X$ operator and corresponding with the permutation $(0,\cdots,p-1)\in S_p$, as follows:
    \begin{align*}
        M=S_\xi X^b&\mapsto \left\lbrace
        \begin{array}{ll}
        \phi^{T_{(p)}}_H(S_\xi)      &   \text{ if $b=0$}\\
        \phi^{T_{(p)}}_H(S_\xi)X     &   \text{ if $b=1$}\\
        \phi^K_H(M^{b_i^{-1}})^{b_i}            &   \text{ if $1<b\leq p-1$}
        \end{array}
        \right.
        \end{align*}
\end{definition}

The map $\phi^K_H$ in Def.~\ref{def: phi^K_H} is very similar to $\phi_1$ in \cite[Def.~2]{FrembsChungOkay2022}. As a result, restricted to $H\cap T$,
the following theorem essentially follows from \cite[Thm.~1]{FrembsChungOkay2022}. For completeness, we state its proof in App.~\ref{app: old lemmata}.

\begin{theorem}{\label{thm:K to H}}
    Let $M,M'\in K$ such that $M^p=M'^p=\one$ and $[M,M']\in Z(K)$, then $\phi^K_H(MM')=\phi^K_H(M)\phi^K_H(M')$.
\end{theorem}

\begin{proof}
    See App.~\ref{app: old lemmata}.
\end{proof}

\begin{definition}\label{def: phi^L_H}
    Let $\phi^L_{H}:L\to H$ be the map defined by its action on elements $M=(\xi,\sigma)\in L$ as follows:
    \[
    \phi^L_H(\xi,\sigma)
    =\begin{cases}
        \phi^{T_{(p)}}_H(\tau\cdot (\xi,\id))  &  \text{ if $\sigma=\id$ and $\exists\  \tau\in Im(\phi^N_{S_p})$ such that $\tau\cdot(\xi,\id)\in H$} \\
        \mf{\phi^{T_{(p)}}_H(\xi,\id)} &  \text{ if $\sigma=\id$ and $\not\exists\tau \in Im(\phi^N_{S_p})$ such that $\tau\cdot(\xi,\id)\in H$ }\\
        \phi^K_H(\phi^N_{S_p}(\xi,\sigma)\cdot (\xi,\sigma) )&  \text{ if $\sigma\neq\id$ and $\sigma^p=\id$}\\
        1&  \text{ otherwise}
    \end{cases}
    \]
\end{definition}

The next lemma shows that $\phi^L_H$ is well-defined, that is, acts consistently on elements $T_{(p)}<L$.

\begin{lemma}{\label{lem: action invariant}}
    Let $M\in T_{(p)}$ and $\tau\in Im(\phi^N_{S_p})$. Then the following relations hold:
    \begin{enumerate}
        \item $\phi^{T_{(p)}}_H(\tau\cdot M)=\phi^{T_{(p)}}_H(M)$
        \item $\phi^L_H(M)=\phi^{T_{(p)}}_H(M)=\phi^K_H(M)$
        \item $\phi^L_H(\tau\cdot M)=\phi^L_H( M)$
    \end{enumerate}
\end{lemma}

\begin{proof}
    (1): Let $M=M_\xi=diag(w^{x_0},...,w^{x_{p-1}})\in T_{(p)}$ where \mf{$x_q:=\sum_{j=0}^{p-1} a_jq^j$}. For $q=0$ and $q=1$, we find $x_0=a_0$ and $x_1=a_0+\cdots+a_{p-1}$, hence, $x_1-x_0=a_1+\cdots+a_{p-1}$. By Def.~\ref{def: phi^L_H}, $\phi^{T_{(p)}}_H(M_\xi)=M_{\chi}$ \mf{where $\chi(q)=\omega^{a_0+\left(\sum_{j=1}^{p-1}a_j\right)q}=\omega^{x_0+(x_1-x_0)q}$}. Clearly, $M_\chi$ thus remains unchanged under the action of $\mathrm{Im}(\phi^N_{S_p})$.
    
    (2): Since $M=(\xi,\id)\in T_{(p)}$, this follows immediately from (i) for the first case in Def.~\ref{def: phi^L_H} (that is, if there exists $\tau\in \mathrm{Im}(\phi^N_{S_p})$ with $\tau\cdot(\xi,\sigma)\in H$), and trivially in the second case. \mf{Moreover, by Def.~\ref{def: phi^K_H}, $\phi^K_H(M)=\phi^{T_{(p)}}_H(M)$.}
    
    (3): Combining (1) and (2), we have $\phi^L_H(\tau\cdot M)=\phi^{T_{(p)}}_H(\tau\cdot M)=\phi^{T_{(p)}}_H(M)=\phi^L_H(M)$.
\end{proof}

\begin{rmk}
    $\phi^{T_{(p)}}_H$ in Def.~\ref{def: phi^K_H} replaces our previous `projection maps' $P$ and $R$ in \cite[Def.~2]{FrembsChungOkay2022}. Indeed, if we use $P$ (or $R$), Lm.~\ref{lem: action invariant} (1) will not hold. For example, consider $S_{\xi_1},S_{\xi_2}\in T_{(5)}$ with $\xi_1(q)=\omega^{q+2q^2+3q^3}$ and $\xi_2(q)=\omega^{3q+4q^2+4q^3}$, that is, $S_{\xi_1}=diag(1,\omega,\omega^4,\omega^2,\omega^3)$ and $S_{\xi_2}=diag(1,\omega,\omega^4,\omega^3,\omega^2)$. Observe that $(4\, 5)\cdot S_{\xi_1}=S_{\xi_2}$ where $(4\, 5)\in Im(\phi^{\tilde{K}}_{S_5})$. By definition of $P$, we compute $P(\xi_1)(q)=\omega^q$, yet $P(\xi_2)(q)=\omega^{3q}$.
\end{rmk}

\begin{proposition}{\label{prop:tp to h, homo, z_hat invariant}}
    $\phi^{T_{(p)}}_H$ is a homomorphism and $\phi^{T_{(p)}}_H(M)=M$ for all $M\in H\cap T_{(p)}$.
\end{proposition}

\begin{proof}
    Let $M_\xi,M_{\xi'}\in T_{(p)}$ and thus $\xi(q)=\omega^{\sum_{j=0}^{p-1} a_jq^j}$,  $\xi'(q)=\omega^{\sum_{j=0}^{p-1} a'_jq^j}$ for $a_j,a'_j\in\Z_p$. Consequently, we compute $\xi(q)\xi'(q)=\omega^{\sum_{j=0}^{p-1}(a_j+a'_j)q^j}$ and thus $\phi^{T_{(p)}}_H(M_\xi M_{\xi'})=\phi^{T_{(p)}}_H(M_\xi)\phi^{T_{(p)}}_H(M_{\xi'})$. 
    
    If $M_\xi\in H \cap T$, then $\xi(q)=\omega^{a+bq}$, where $a,b\in\Z_p$. By Def.~\ref{def: phi^L_H}, $\phi^{T_{(p)}}_H(M_\xi)=M_\xi$.
\end{proof}

\begin{lemma}{\label{lem: s_chi pull out}}
     Let $M\in L$ and $S_\chi\in H\cap T_{(p)}$, then
     $\phi^L_H(S_\chi M)=S_\chi\phi^L_H(M)$.
\end{lemma}
\begin{proof}
    If $M\in T_{(p)}$ this holds by Prop.~\ref{prop:tp to h, homo, z_hat invariant}. If $M\not\in T_{(p)}$, let $\tau=\phi^N_{S_p}(S_\chi M)$. By Lm.~\ref{lem: action invariant} (3), Prop.~\ref{prop:tp to h, homo, z_hat invariant} and Thm.~\ref{thm:K to H},
    \begin{align*}
        \phi^L_H(S_\chi M)
        =\phi^K_H(\tau\cdot (S_\chi M)))
        &=\phi^K_H((\tau\cdot S_\chi)(\tau\cdot M))\\
        &=\phi^K_H(\tau\cdot S_\chi)\phi^K_H(\tau\cdot M)
        =\phi^{T_{(p)}}_H(\tau\cdot S_\chi)\phi^L_H(M)
        =S_\chi\phi^L_H(M)\; .\qedhere
    \end{align*}
\end{proof}

\begin{theorem}{\label{thm: main L to H}}
    Let $M,M'\in L$ such that $[M,M']\in Z(K)$ and $M^p_1=M^p_2=\one$, then
    $$\phi^L_H(MM')=\phi^L_H(M)\phi^L_H(M')\; .$$
\end{theorem}

\begin{proof}
    We have the following cases to consider:
    \begin{enumerate}
    \item $M,M'\in T_{(p)}$
    \item $M\in T_{(p)}$ and $M'\not\in T_{(p)}$ (analogously, $M\not\in T_{(p)}$ and $M'\in T_{(p)}$)
    \item $M,M'\not\in T_{(p)}$
    \end{enumerate}
    (1): Since $M,M'\in T_{(p)}$, also $MM'\in T_{(p)}$. Hence, by Lm.~\ref{lem: action invariant} and Prop.~\ref{prop:tp to h, homo, z_hat invariant}, we find
    \begin{align*}
        \phi^L_H(MM')
        =\phi^{T_{(p)}}_H(MM')
        =\phi^{T_{(p)}}_H(M)\phi^{T_{(p)}}_H(M')
        =\phi^L_H(M)\phi^L_H(M')\; .
    \end{align*}
    (2): Let $\tau=\phi^N_{S_p}(M')$ such that $\tau\cdot M'\in H$. By Lm.~\ref{lem: comm_tildeK}, also $\tau\cdot M\in H\cap T_{(p)}\leq H$. By Lm.~\ref{lem: action invariant} and Thm.~\ref{thm:K to H},
    \begin{align*}
        \phi^L_H(MM')
        &=\phi^K_H(\tau\cdot (MM'))\\
        &=\phi^K_H((\tau\cdot M)(\tau\cdot M'))\\
        &=\phi^K_H(\tau\cdot M)\phi^K_H(\tau\cdot M')\\
        &=\phi^{T_{(p)}}_H(\tau\cdot M)\phi^L_H(M')\\
        &=\phi^{T_{(p)}}_H(M)\phi^L_H(M')\\
        &=\phi^L_H(M)\phi^L_H(M')\; .
    \end{align*}
    (3): By Lm.~\ref{lem: comm_tildeK}, $\phi^N_{S_p}(M)=\phi^N_{S_p}(M')=:\tau$. If $MM'\not\in T_{(p)}$, then also $\tau=\phi^N_{S_p}(MM')$, hence,
    \begin{align*}
        \phi^L_H(MM')
        =\phi^K_H(\tau\cdot(MM'))=\phi^K_H((\tau\cdot M)(\tau\cdot M'))
        =\phi^K_H(\tau\cdot M)\phi^K_H(\tau\cdot M')
        =\phi^L_H(M)\phi^L_H(M')\; .
    \end{align*}
    If $MM'\in T_{(p)}$, then applying Lm.~\ref{lem: action invariant} (1) and Thm.~\ref{thm:K to H} again, we obtain
    \begin{align*}
        \phi^L_H(MN)
        =\phi^L_H(\tau\cdot (MN))
        =\phi^K_H(\tau\cdot (MN))
        &=\phi^K_H((\tau\cdot M)(\tau\cdot N))\\
        &=\phi^K_H(\tau\cdot M)\phi^K_H(\tau\cdot N)
        =\phi^L_H(M)\phi^L_H(N)\; .\qedhere
    \end{align*}
\end{proof}

Finally, we extend $\phi^L_H$ to multiple tensor copies, $\phi^{L^{\otimes n}}_{H^{\otimes n}}:L^{\otimes n}\to H^{\otimes n}$, defined by $\bigotimes_{i=1}^nM_i\mapsto \bigotimes_{i=1}^n\phi^L_H(M_i)$. Note that $\phi^{L^{\otimes n}}_{H^{\otimes n}}$ is well-defined: if $M_1\otimes\cdots\otimes M_n=M'_1\otimes \cdots\otimes M'_n$ and thus $M'_i=c_iM_i$ for $c_i\in Z(L)$, $i\in[n]$ with $\prod_{i=1}^n c_i=1$, then one quickly checks with Lm.\ref{lem: s_chi pull out} that
\begin{align*}
    \phi^{L^{\otimes n}}_{H^{\otimes n}}(M'_1\otimes \cdots \otimes M'_n)
    &=\phi^L_H(c_1M_1)\otimes\cdots\otimes \phi^L_H(c_nM_n)\\
    &=c_1\phi^L_H(M_1)\otimes\cdots\otimes c_n\phi^L_H(M_n)
    =\phi^{L^{\otimes n}}_{H^{\otimes n}}(M_1\otimes\cdots\otimes M_n)\; .\qedhere
\end{align*}

\begin{theorem}{\label{thm: L to H}}
    Let $M,M'\in L^{\otimes n}$ such that $[M,M']=\one$ and $M^p=M'^p=\one$, then
    \begin{align*}
        \phi^{L^{\otimes n}}_{H^{\otimes n}}(MM'
        )
        =\phi^{L^{\otimes n}}_{H^{\otimes n}}(M)\phi^{L^{\otimes n}}_{H^{\otimes n}}(M')\; .
    \end{align*}
\end{theorem}

\begin{proof}
    Let $M=\bigotimes_{i=1}^nM_i$ and $M'=\bigotimes_{i=1}^nM'_i$. Since $[M,M']=\one$, we have $[M_i,M'_i]\in Z(L)$. By Thm.~\ref{thm: main L to H},
    \begin{align*}
        \phi^{L^{\otimes n}}_{H^{\otimes n}}(MM')
        &=\bigotimes \phi^{L}_{H}(M_iM'_i)\\
        &=\bigotimes \left(\phi^{L}_{H}(M_i)\phi^{L}_{H}(M'_i)\right)\\
        &=\left(\bigotimes \phi^{L}_{H}(M_i)\right)\left(\bigotimes \phi^{L}_{H}(M'_i)\right)\\
        &=\phi^{L^{\otimes n}}_{H^{\otimes n}}(M)\phi^{L^{\otimes n}}_{H^{\otimes n}}(M')\; .\qedhere
    \end{align*}
\end{proof}

Thm.~\ref{thm: L to H} immediately implies the following special case of Thm.~\ref{thm: FCO-2}.

\begin{corollary}\label{cor: no qsol for L}
    Let $\Gamma$ be the solution group of a LCS over $\zz_p$ for $p$ odd prime. Then the LCS admits a quantum solution $\eta: \Gamma \ra L^{\otimes n}$ if and only if it admits a classical solution $\eta: \Gamma \ra \zz_p$.
\end{corollary}

\begin{proof}
    Since $\phi^{L^{\otimes n}}_{H^{\otimes n}}$ is a homomorphism on $p$-torsion abelian subgroups by Thm.~\ref{thm: L to H}, a LCS admits a quantum solution $\eta:\Gamma\ra L^{\otimes n}$ if and only if it admits a quantum solution $\eta:\Gamma \ra H^{\otimes n}$ (see \cite[Cor.~4]{FrembsChungOkay2022}). Yet, every quantum solution $\eta:\Gamma\ra H^{\otimes n}$ is classical by \cite[Thm.~2]{QassimWallman2020}.
\end{proof}

\subsection{A homomorphism in $p$-torsion abelian subgroups $\phi^{N_Q^{\otimes n}}_{H^{\otimes n}}:N_Q^{\otimes n}\ra H^{\otimes n}$}

\mf{One may hope to use Thm.~\ref{thm: L to H} towards a proof of Thm.~\ref{thm: FCO-2}. This strategy works under the restriction to $T_{(p)}\leq Q\leq T_{(p^m)}$ for $m\in\mathbb{N}$, yet runs into problems beyond that.} We will write $N_Q\leq N$ for the group in
\begin{align}\label{ses: N_Q}
    1 \ra Q \ra N_Q \ra S_p \ra 1\; .
\end{align}

\begin{definition}\label{def: q homomorphism}
    Let $q:\zz_{p^m}\cong \left\langle exp\left(\frac{2\pi i}{p^m}\right)\right\rangle \to \zz_p\cong \left\langle exp\left(\frac{2\pi i}{p}\right)\right\rangle$ be the homomorphism defined by
    \begin{align*}
        exp\left(\frac{2n\pi i}{p^m}\right)&\mapsto exp\left(\frac{2n\pi i}{p}\right)\; .
    \end{align*}
    Then define the following maps on $N_Q$ in Eq.~(\ref{ses: N_Q}):
    \begin{align*}
        \phi^{N_Q}_{L}:N_Q&\to L &
        (diag(x_1,\cdots,x_p),\sigma)&\mapsto (diag(q(x_1),\cdots,q(x_p)),\sigma) \\
        \phi^{N_Q^{\otimes n}}_{L^{\otimes n}}:N_Q^{\otimes n}&\to L^{\otimes n} &
        \bigotimes_{i=1}^n M_i&\mapsto \bigotimes_{i=1}^n \phi^{N_Q}_{L}(M_i)
    \end{align*}
\end{definition}

\begin{lemma}\label{lm: reduction to L^n}
    $\phi^{N_Q}_{L}$ and $\phi^{N_Q^{\otimes n}}_{L^{\otimes n}}$ are homomorphisms.
\end{lemma}

\begin{proof}
    First, consider the map $\phi^{N_Q}_{L}$, and let $M=(diag(m_1,\cdots,m_p),\sigma)$ and $M'=(diag(m'_1,\cdots,m'_p),\tau)$. Then
    $$MM'=(diag(m_1+m'_{\sigma(1)},\cdots,m_p+m'_{\sigma(p)}),\sigma\tau)\; ,$$
    and since $q$ is a homomorphism, we have
    \begin{align*}
        \phi^{N_Q}_{L}(MM')
        &=\left(diag(q(m_1+m'_{\sigma(1)}),\cdots,q(m_p+m'_{\sigma(p)})),\sigma\tau\right)\\
        &=\left(diag(q(m_1)+q(m'_{\sigma(1)}),\cdots,q(m_p)+q(m'_{\sigma(p)})),\sigma\tau\right)\\
        &=\left(diag(q(m_1),\cdots,q(m_p)),\sigma)(diag(q(m'_1),\cdots,q(m'_p)),\tau\right)\\
        &=\phi^{N_Q}_{L}(M)\phi^{N_Q}_{L}(M')\; .
    \end{align*}
    Finally, it is easy to see that $\phi^{N_Q^{\otimes n}}_{L^{\otimes n}}$ is a homomorphism since $\phi^{N_Q}_{L}$ is.
\end{proof}

Lm.~\ref{lm: reduction to L^n} immediately implies the following special case of Thm.~\ref{thm: FCO-2}.

\begin{corollary}
    \mf{If $\eta:\Gamma\to N_Q^{\otimes n}$ is a quantum solution, then $\phi^{N_Q^{\otimes n}}_{L^{\otimes n}}\circ \eta:\Gamma\to L^{\otimes n}$ is a quantum solution.}
\end{corollary}

\begin{proof}
    This follows from Lm.~\ref{lm: reduction to L^n} and Cor.~\ref{cor: no qsol for L} and \cite[Cor.~4]{FrembsChungOkay2022}.
\end{proof}

Comparing with the proof of Thm.~\ref{thm: FCO-2} in App.~\ref{app: proof of main theorem}, the homomorphism $q:T_{(p^m)}\ra T_{(p)}$ in Def.~\ref{def: q homomorphism} enforces additional constraints on  to choice of generators $(\xi,\sigma)$ in the respective subgroups $A(\xi,\sigma)$ with $S_\xi\in T_{(p^m)}$ for $m\in\mathbb{N}$. While these are natural under the restriction to $T_{(p^m)}$, demanding $\phi^{N_Q}_L$ to be a homomorphism on $Q$ is more than we need. The generalisation to $Q\leq T$ in Thm.~\ref{thm: FCO-2} is achieved by removing these constraints.
\section{Proof of Thm.~\ref{thm:K to H}}\label{app: old lemmata}

In this section, we provide a proof for Thm.~\ref{thm:K to H}, which is based on slight modifications of Lm.~2, Lm.~7 and Lm.~8 from Ref.~\cite{FrembsChungOkay2022}. We re-state them here with (repeated, but adapted) proofs for completeness.

\begin{lemma}{\label{lem: xi,sigma to power n}}
    For every $(\xi,\sigma)\in K$ and $n\in\mathbb{N}$, it holds that $(\xi,\sigma)^n=\left(\sum_{i=0}^{n-1}\sigma^i\cdot\xi,\sigma^n\right)$.
\end{lemma}

\begin{proof}
    We proceed by induction on $n$. The case $n=1$ is trivial, and the inductive step $k\ra k+1$ follows from
    \begin{align*}
        \hspace{3cm}
        (\xi,\sigma)^{k+1}
        =\left(\sum_{i=0}^{k-1}\sigma^i\cdot\xi,\sigma^k\right)(\xi,\sigma)
        =\left(\sigma^k\cdot\xi+\sum_{i=0}^{k-1}\sigma^i\cdot\xi,\sigma^{k+1}\right)
        =\left(\sum_{i=0}^k\sigma^i\cdot\xi,\sigma^{k+1}\right)\; .\hspace{3cm}\qedhere
    \end{align*}
\end{proof}

\vspace{.2cm}

\begin{lemma}{\label{lem: pull chi out in K}}
    Let $M=S_\xi X^b=(\xi,\sigma^b)\in K$ and $S_\chi=(\chi,1)\in H\cap T$, then
    \begin{align*}
        \phi^K_H(S_{\chi}M)=\phi^K_H(S_{\chi})\phi^K_H(M)=S_\chi \phi^K_H(M)\; .
    \end{align*}
\end{lemma}

\begin{proof}
    Let $y\in\mathbb{N}$ such that $\sigma^{yb}=\sigma$, then
    $$(\chi,\sigma^b)=((\chi,\sigma^b)^y)^b=\left(\sum_{i=1}^{y-1}\sigma^{bi}\cdot\chi,\sigma\right)^b=\left(\sum_{j=0}^{b-1}\sigma^j\cdot\left((\sum_{i=0}^{y-1}\sigma^{bi}\cdot\chi\right),\sigma^b\right)\; .$$
    Since $\phi^{T_{(p)}}_H(S_\chi)=S_\chi$ by Prop.~\ref{prop:tp to h, homo, z_hat invariant} and $\sigma\cdot S_\chi\in H\cap T$, we have $\phi^{T_{(p)}}_H(\sigma\cdot\chi)=\sigma\cdot\chi=\sigma\cdot\phi^{T_{(p)}}_H(\chi)$. With Prop.~\ref{prop:tp to h, homo, z_hat invariant} again,
    $$
    \chi=\sum_{j=0}^{b-1}\sigma^j\cdot\left(\sum_{i=0}^{y-1}\sigma^{bi}\cdot\chi\right)=\sum_{j=0}^{b-1}\sigma^j\cdot\phi^{T_{(p)}}_H\left(\sum_{i=0}^{y-1}\sigma^{bi}\cdot\chi\right)
    $$
    The statement now follows from the following computation, which uses Lm.~\ref{lem: xi,sigma to power n} and Prop.~\ref{prop:tp to h, homo, z_hat invariant} repeatedly,
    \begin{align*}
        \phi^K_H(S_\chi M)
        &=\phi^K_H(\chi+\xi,\sigma^b)\\
        &=\phi^K_H(((\chi+\xi,\sigma^b)^y)^b)\\
        &=\phi^K_H\left(\left(\sum_{i=0}^{y-1}\sigma^{bi}\cdot\chi+\sum_{i=0}^{y-1}\sigma^{bi}\cdot\xi,\sigma\right)^b\right)\\
        &=\left(\phi^{T_{(p)}}_H\left(\sum_{i=0}^{y-1}\sigma^{bi}\cdot\chi\right)+\phi^{T_{(p)}}_H\left(\sum_{i=0}^{y-1}\sigma^{bi}\cdot\xi\right),\sigma\right)^b\\
        &=\left(\sum_{j=0}^{b-1}\sigma^j\cdot\phi^{T_{(p)}}_H\left(\sum_{i=0}^{y-1}\sigma^{bi}\cdot\chi\right)+\sum_{j=0}^{b-1}\sigma^j\cdot\phi^{T_{(p)}}_H\left(\sum_{i=0}^{y-1}\sigma^{bi}\cdot\xi\right),\sigma^b\right)\\
        &=\left(\chi+\sum_{j=0}^{b-1}\sigma^j\cdot\phi^{T_{(p)}}_H\left(\sum_{i=0}^{y-1}\sigma^{bi}\cdot\xi\right),\sigma^b\right)\\
        &=(\chi,1)\left(\sum_{j=0}^{b-1}\sigma^j\cdot\phi^{T_{(p)}}_H\left(\sum_{i=0}^{y-1}\sigma^{bi}\cdot\xi\right),\sigma^b\right)\; .
    \end{align*}
    On the other hand, ones equally computes
    \begin{align*}
        \phi^K_H(S_\chi)\phi^K_H(M)
        &=(\chi,1)\phi^K_H(((\xi,\sigma^b)^y)^b)\\
        &=(\chi,1)\phi^K_H\left(\left(\sum_{i=0}^{y-1}\sigma^{bi}\cdot\xi,\sigma\right)^b\right)\\
        &=(\chi,1)\left(\phi^{T_{(p)}}_H\left(\sum_{i=0}^{y-1}\sigma^{bi}\cdot\xi\right),\sigma\right)^b\\
        &=(\chi,1)\left(\sum_{j=0}^{b-1}\sigma^j\cdot\phi^{T_{(p)}}_H\left(\sum_{i=0}^{y-1}\sigma^{bi}\cdot\xi\right),\sigma^b\right)\\
        &=\phi^K_H(S_\chi M)\; .\qedhere
    \end{align*}
\end{proof}

\begin{lemma}{\label{lem: M,N not in T}}
    Let $M,M'\in K$ with $M,M'\not\in T_{(p)}$ and $[M,M']=\omega^c\one$ for $c\in\Z_p$, then $\phi^K_H(MM')=\phi^K_H(M)\phi^K_H(M')$.
\end{lemma}

\begin{proof}
    By Lm.~\ref{lem: xi,sigma to power n}, without loss of generality, we assume $N=(\xi,\sigma)^b$. By \cite[Lm.~4 (3)]{FrembsChungOkay2022}, there exists $\chi\in\hat{\Z}_p$ and $y\in \Z_p$ such that $M=S_\chi N^y=S_\chi(\xi,\sigma)^{yb}$. By Lm.~\ref{lem: s_chi pull out}, we have
    \begin{align*}
        \phi^K_H(MM')
        =\phi^K_H(S_\chi(\xi,\sigma)^{yb}(\xi,\sigma)^b)
        =\phi^K_H(S_\chi(\xi,\sigma)^{yb+b})
        =S_\chi(\phi^{T_{(p)}}_H(\xi),\sigma)^{yb+b}\; ,
    \end{align*}
    On the other hand, we have
    \begin{align*}
        \hspace{.75cm}
        \phi^K_H(M)\phi^K_H(M')
        =\phi^K_H(S_\chi(\xi,\sigma)^{yb})\phi^K_H((\xi,\sigma)^b)
        =S_\chi(\phi^{T_{(p)}}_H(\xi),\sigma)^{yb}(\phi^{T_{(P)}}_H(\xi),\sigma)^b
        =S_\chi(\phi^{T_{(p)}}_H(\xi),\sigma)^{yb+b}\; .
        \hspace{1cm}\qedhere
    \end{align*}
\end{proof}

\vspace{.2cm}

\begin{proof}[Proof of Thm.~\ref{thm:K to H}]
    There are three cases to consider:
    \begin{itemize}
        \item Case 1: $M,M'\in T_{(p)}$. Follows from Prop.~\ref{prop:tp to h, homo, z_hat invariant}.
        \item Case 2: $M\in T_{(p)}$ and $M'\not\in T_{(p)}$. By \cite[Lm.~4]{FrembsChungOkay2022}, $M\in H\cap T$, hence, the result follows from
        \begin{align*}
            \phi^K_H(MM')
            =\phi^K_H(\omega^\alpha M'M)
            =\omega^\alpha \phi^K_H(M')\phi^K_H(M)
            =\omega^\alpha\omega^{-\alpha}\phi^K_H(M)\phi^K_H(M')
            =\phi^K_H(M)\phi^K_H(M')\; ,
        \end{align*}
        where we used $[M,M']\in Z(K)$ in the first step, Lm.~\ref{lem: pull chi out in K} in the second, and $[M,M']=[\phi^K_H(M),\phi^K_H(M')]$ as follows from \cite[Cor.~3]{FrembsChungOkay2022} in the third. The case $M\not\in T_{(p)}$ and $M'\in T_{(p)}$ is analogous.
        \item Case 3: $M,M'\notin T_{(p)}$. Then result follows from Lm.~\ref{lem: M,N not in T}.\hfill\qedhere
    \end{itemize}
\end{proof}

\end{document}